\documentclass[a4paper,11pt]{article}
\usepackage{jheppub}
\usepackage{amsmath,amssymb,amsthm,mathrsfs}
\usepackage{booktabs}
\usepackage{etoolbox}
\hypersetup{pdftitle={Mirror Symmetry, the Schwarzian Derivative, and the Modularity of Seiberg-Witten Theory},pdfauthor={Wenzhe Yang}}
\numberwithin{equation}{section}
\newtheorem{theorem}{Theorem}[section]
\newtheorem{proposition}[theorem]{Proposition}
\newtheorem{lemma}[theorem]{Lemma}

\theoremstyle{definition}

\theoremstyle{remark}
\newtheorem{remark}[theorem]{Remark}
\newcommand{\HH}{\mathbb H}
\newcommand{\ZZ}{\mathbb Z}
\newcommand{\CC}{\mathbb C}
\newcommand{\RR}{\mathbb R}
\newcommand{\PP}{\mathbb P}
\newcommand{\Sch}[2]{\left\{#1,#2\right\}}
\newcommand{\Yrig}{C^{\mathrm{rig}}}
\newcommand{\YSW}{\mathcal Y_{\mathrm{SW}}}
\makeatletter
\gdef\@fpheader{}
\makeatother
\title{\boldmath Mirror Symmetry, the Schwarzian Derivative, and the Modularity of Seiberg--Witten Theory}
\author{Wenzhe Yang}
\affiliation{Sang Shu Wa Dong Zhuang, Lincheng, Xingtai, Hebei, China, 054300}
\emailAdd{wenzheyang87@gmail.com}
\newcommand{\Span}{\operatorname{span}}
\newcommand{\Tens}{\mathscr T}
\newcommand{\Quad}{\mathcal Q}
\newcommand{\Ygen}{\mathcal Y}
\let\originaltableofcontents\tableofcontents
\makeatletter
\renewcommand{\tableofcontents}{\begingroup\small
\setlength{\baselineskip}{11.5pt}
\patchcmd{\l@section}{1.0em \@plus\p@}{0.1em \@plus\p@}{}{}
\originaltableofcontents\endgroup}
\makeatother
\renewcommand{\afterTocRuleSpace}{\clearpage}
\apptocmd{\thebibliography}{\setlength{\itemsep}{2pt plus 0.2ex minus 0.05ex}}{}{}
\abstract{We study the Schwarzian derivative of the Seiberg--Witten period ratio in the rigid limit of the two-parameter mirror family associated with a degree-$12$ Calabi--Yau hypersurface in $\PP^4_{1,1,2,2,6}$. The leading tangential Yukawa coefficient is proportional to the holomorphic cubic tensor of rigid special K\"ahler geometry. Using the mixed normal components and a prescribed constant-coefficient connection rule, we reconstruct the Schwarzian quadratic differential, whose coefficient in the effective coupling $\tau$ is $\pi^2\theta_2(\tau/2)^8/8$. With the weak-coupling frame, dynamical scale, and perturbative normalization fixed, we prove that the holomorphic germ of this coefficient determines the prepotential up to an additive constant within a convergent local class, and vanishes identically if and only if all instanton corrections vanish. We then use the genus-zero mirror correspondence in the chosen chamber to relate the first two Schwarzian Fourier coefficients to compact Gromov--Witten invariants. The Picard--Fuchs system and full Yukawa tensor give the generating functions at base degrees one and two and their asymptotics as the fiber degree tends to infinity. Locally uniform period estimates fix the rigid-limit normalization and justify extraction of the gauge-instanton coefficients. This yields an explicit asymptotic relation between the compact invariants and the normalized second Schwarzian coefficient. The extension to higher base degree remains conditional.}

\begin{document}
\maketitle
\raggedbottom

\section{Introduction}
We use the Schwarzian derivative of the Seiberg--Witten period ratio to study the relation between gauge theory and the geometry of a compact Calabi--Yau threefold. In the rigid limit of the two-parameter mirror family, we reconstruct the associated quadratic differential from the Yukawa tensor using a prescribed connection. The holomorphic germ of its coefficient in the effective coupling determines the normalized prepotential up to an additive constant. We relate the first two Fourier coefficients of this germ to the growth of the compact genus-zero Gromov--Witten invariants as the fiber degree tends to infinity.

The exact solution of pure $N=2$ supersymmetric Yang--Mills theory with gauge group $\mathrm{SU}(2)$ expresses the electric and magnetic special coordinates, $a$ and $a_D$, as periods of the Seiberg--Witten differential on a family of elliptic curves \cite{SWTheory,Lerche}. If $\mathcal F$ denotes the holomorphic prepotential and $u$ is the coordinate on the Coulomb branch, set
\begin{equation}\label{eq:basic}
 r=\frac{a_D}{a},\qquad
 \tau=\frac{da_D}{da}=\frac{d^2\mathcal F}{da^2}.
\end{equation}
The ratio $r$ uses periods of the meromorphic Seiberg--Witten differential, whereas the effective complexified gauge coupling $\tau$ uses periods of the holomorphic differential representing its Gauss--Manin derivative. Each ratio defines a projective connection through its Schwarzian derivative. We write $\Sch{f}{z}$ for the Schwarzian of a locally univalent function $f$, as defined in Section~\ref{sec:schwarzian}. The difference of these projective connections gives the quadratic differential
\begin{equation}\label{eq:Qintro}
 \Quad=\bigl(\Sch{r}{u}-\Sch{\tau}{u}\bigr)(du)^2
 =\Sch{r}{\tau}(d\tau)^2.
\end{equation}
We recover $\Quad$ from the Yukawa tensor of a degenerating Calabi--Yau threefold, together with the connection rule specified below, and relate its coefficient in the variable $\tau$ to genus-zero invariants of the compact geometry.

A theta-function expression for the Seiberg--Witten Schwarzian was obtained in \cite[Section~5.3]{YangSWLambda}, with a different integral cycle marking and coupling coordinate. We rederive this identity directly from the Picard--Fuchs equation in the primitive cycle normalization of the quartic family.

For the geometric comparison, we use the two-parameter mirror family of the resolved degree-$12$ hypersurface in $\PP^4_{1,1,2,2,6}$. Its point-particle limit recovers the Seiberg--Witten periods from suitable combinations of the six Calabi--Yau periods \cite{Candelas,Kachru}. Two issues arise in this comparison. The first nonzero tangential Yukawa term is a cubic differential, whereas $\Quad$ is quadratic, and the tangential coefficient depends on which transverse coordinate is held fixed. We address both issues by retaining the mixed normal components of the leading Yukawa tensor. These select a horizontal splitting on the open exceptional divisor of the double cover, including at $u=0$. With a prescribed flat connection on the coefficient line bundle, two local primitives then give projective connections whose difference is $\Quad$. The resulting quadratic differential is independent of the Hodge and normal frames.

We set the dynamical scale to one and use periods over the primitive cycles on the quartic elliptic curve specified in Section~\ref{sec:sw}. The smooth base and the auxiliary variables are
\begin{equation}\label{eq:coordinates}
 B=\PP^1_u\setminus\{-1,1,\infty\},\qquad
 \varphi=u^{-2},\qquad t=\frac{\tau}{2},\qquad
 q=e^{\pi i\tau}=e^{2\pi it}.
\end{equation}
The variable $\varphi$ is a coordinate on the quotient by $u\mapsto-u$; it is not a local coordinate at $u=0$, where we use $u$ itself. Throughout the paper, $\lambda$ denotes the Legendre parameter or the modular lambda function. The relation between these variables is $\lambda(t)=2/(u+1)$ and $\varphi=\lambda^2/(2-\lambda)^2$. The Picard--Fuchs equation and Legendre uniformization give
\begin{equation}\label{eq:mainintro}
 \Quad=-\frac{2(du)^2}{(u^2-1)^2},\qquad
 \Sch{r}{\tau}
 =\frac{\pi^2}{8}\theta_2(t)^8
 =32\pi^2E_{\infty,4}(t),\qquad
 E_{\infty,4}(t)=\frac{E_4(t)-E_4(2t)}{240}.
\end{equation}
Here $\theta_2$ is a Jacobi theta constant and $E_4$ is the normalized Eisenstein series of weight four, both defined in Section~\ref{sec:schwarzian}. The theta series uses the nome $e^{\pi it}$, while our Fourier variable is $q=e^{2\pi it}$. The form $E_{\infty,4}$ is holomorphic for $\Gamma_0(2)$ and vanishes at the cusp $\infty$; because it is nonzero at the cusp $0$, it is not a cusp form. The modular identity follows directly from the differential equation.

Near the rigid-limit divisor, we use the mirror coordinates
\begin{equation}\label{eq:chartintro}
 x=1+\epsilon u,\qquad y=\frac{\epsilon^2}{(1+\epsilon u)^2},
\end{equation}
where $\epsilon$ is transverse to the divisor and $u$ is tangential. Let $\Yrig$ be the coefficient of $\epsilon$ in the tangential Yukawa coupling, with $\epsilon$ held fixed and with the algebraic Hodge normalization specified in Section~\ref{sec:mirror}. We prove
\begin{equation}\label{eq:cubicintro}
 \Yrig=\frac{2}{\pi^2}\YSW,
 \qquad \YSW=(da)^2d\tau.
\end{equation}
To pass from this cubic differential to $\Quad$, we use the mixed normal components to choose a horizontal complement along the open exceptional divisor $B$. Including $u=0$ is essential to the uniqueness of this choice. We then equip the coefficient line bundle with the connection prescribed by the constant-coefficient rule \eqref{g:connection}. A general identity for rank-two flat symplectic bundles shows that the resulting quadratic differential is exactly \eqref{eq:Qintro}.

The Schwarzian also determines the instanton part of the low-energy theory. For the normalized convergent class of Section~\ref{sec:quantum}, the holomorphic germ $S(q)=\Sch{r}{\tau}$ determines the prepotential up to an additive constant once the weak-coupling electric-magnetic frame, dynamical scale, and quadratic normalization of \eqref{eq:F} are fixed. In this class, $S$ vanishes identically precisely when every instanton coefficient vanishes. We measure the instanton corrections relative to the one-loop prepotential, which is perturbatively exact. Their reconstruction requires the full holomorphic germ; a value at a single point does not suffice.

The compact geometry enters through its genus-zero invariants. Write $N^0_{j,k}$ for the rational Gromov--Witten invariant in fiber degree $j$ and base degree $k$, and $n^0_{j,k}$ for the instanton numbers obtained by multiple-cover inversion. From the compact Picard--Fuchs system and the full Yukawa tensor, we compute the generating functions at base degrees one and two and prove
\[
 n^0_{j,k}\sim\kappa_k e^{2\pi j}j^{4k-3}\quad(j\to\infty),\qquad k=1,2,
 \qquad \frac{\kappa_2}{\kappa_1^2}=\frac1{384}.
\]
At these base degrees, the multiple-cover correction is exponentially smaller than the leading term, so the same asymptotic ratio holds for $N^0_{j,k}$. Writing
\begin{equation*}
 S(\tau)=\Sch{a_D/a}{\tau}=32\pi^2\sum_{n\ge1}A_nq^n,
\end{equation*}
we obtain
\begin{equation}\label{eq:enumerativeintro}
 \frac{A_2}{A_1^2}
 =46080\lim_{j\to\infty}
 \frac{e^{2\pi j}N^0_{j,2}}{j^3(N^0_{j,1})^2}-112=8.
\end{equation}
The normalization in this formula comes from locally uniform period estimates in a fixed weak-coupling sector. These estimates fix the electric-coordinate and transverse normalization constants. Uniform convergence of the normalized Yukawa coupling then justifies extraction of $f_1=1/32$ and $f_2=5/2^{14}$. The subtraction in \eqref{eq:enumerativeintro} arises because the Schwarzian coefficients depend nonlinearly on the gauge-instanton coefficients. To interpret the compact calculation enumeratively, we take as an input the genus-zero mirror correspondence in the large-radius chamber of Section~\ref{sec:enumerative}. The argument does not require a comparison of the full integral variations of Hodge structure.

\paragraph{Notation and branches.}
We write $\HH=\{z\in\CC:\operatorname{Im}z>0\}$ for the upper half-plane and $\PP^1$ for the complex projective line. Products such as $(da)^2d\tau$ are symmetric tensor products, and $[q^n]F$ denotes the coefficient of $q^n$ in a power series $F$. We continue periods, logarithms, and fractional powers from the branches specified in the text, working on simply connected domains or the corresponding covering spaces. In the compact calculation, the variables $Q=e^{2\pi iT}$ and $Q_2=e^{2\pi i\Sigma}$ are distinct from the gauge-theory variable $q$.

Section~\ref{sec:sw} fixes the period and monodromy conventions used throughout. With these in place, Section~\ref{sec:schwarzian} derives the Schwarzian identities and their modular form, and Section~\ref{sec:quantum} recovers the normalized prepotential. The geometric reconstruction occupies Sections~\ref{sec:mirror} and \ref{sec:comparison}, where we compute the leading logarithmic Yukawa tensor and recover $\Quad$. We then turn to the compact generating functions and their asymptotics in Section~\ref{sec:enumerative}. Section~\ref{sec:enumerative-rigid} completes the comparison by proving the normalized rigid limit and deriving \eqref{eq:enumerativeintro}.

\section{A brief review of Seiberg--Witten theory}\label{sec:sw}

The Schwarzian and rigid-limit comparisons require a consistent normalization of the electric and magnetic periods. We recall the Seiberg--Witten solution \cite{SWTheory} in the conventions of \cite{Lerche}, choosing a primitive integral cycle basis for the quartic family. The monodromy, Picard--Fuchs equation, Wronskians, and weak-coupling expansion then fix the numerical factors and gauge-instanton coefficients used throughout the paper.

\subsection{Special coordinates and period monodromy}
An $N=2$ vector multiplet contains a gauge field, two Weyl fermions, and a complex scalar $\phi$, all taking values in the adjoint representation of $\mathrm{SU}(2)$. Write $\phi=\sum_i\phi_i\sigma_i$, where $\sigma_i$ are the Pauli matrices. Up to a positive normalization constant, the classical scalar potential is
\begin{equation}\label{eq:potential}
 V(\phi)=\operatorname{Tr}\!\left([\phi,\phi^\dagger]^2\right).
\end{equation}
Here $\phi^\dagger$ is the Hermitian adjoint, and the trace is taken in the defining two-dimensional representation. Since the potential vanishes precisely when $[\phi,\phi^\dagger]=0$, the scalar in a classical vacuum is normal and can be diagonalized by a gauge transformation as $a\operatorname{diag}(1,-1)$. The Weyl group identifies $a$ with $-a$, leaving $u=\operatorname{Tr}\phi^2=2a^2$ as the coordinate on the classical Coulomb branch. In the quantum theory, $u=\langle\operatorname{Tr}\phi^2\rangle$ remains gauge invariant, whereas $a(u)$ is a multivalued special coordinate with $u\sim2a^2$ at weak coupling. Setting the dynamical scale to one, we obtain the smooth quantum Coulomb branch
\begin{equation}
 \mathscr M=B=\PP^1_u\setminus\{-1,1,\infty\}.
\end{equation}
The dual special coordinate is $a_D=\partial_a\mathcal F$, where $\mathcal F$ is the prepotential. The effective complexified coupling is \cite{SWTheory,Lerche}
\begin{equation} \label{eq:defntau}
    \tau = \partial_a a_D = \frac{\partial_u a_D}{ \partial_u a}.
\end{equation}
In a local electric-magnetic frame with $\operatorname{Im}\tau>0$, the K\"ahler potential and metric are
\begin{equation}\label{eq:metric}
 K=\operatorname{Im}(\overline a\,a_D),\qquad
 ds^2=\operatorname{Im}\tau\,|da|^2.
\end{equation}

For an electric-magnetic charge pair $(n_e,n_m)\in\ZZ^2$, the central charge is $Z=n_ea+n_ma_D$. The periods are specified by an ordered integral basis of $H_1(\mathscr E_u,\ZZ)$ for the quartic curve \eqref{eq:quartic}, with its intersection orientation fixed. Continuing this basis by parallel transport on a simply connected base domain gives an \emph{integral cycle marking}, or simply a \emph{marking}, and the period vector $\Pi=(a_D,a)^{\mathsf T}$. Analytic continuation around a singularity acts by
\begin{equation}
    \begin{pmatrix}
        a_D \\
        a   \\
    \end{pmatrix}
    \longrightarrow
    M 
    \begin{pmatrix}
        a_D \\
        a   \\
    \end{pmatrix},
\end{equation}
where $M\in\mathrm{SL}_2(\ZZ)$ has the form
\begin{equation}
 M=\begin{pmatrix}\alpha&\beta\\\gamma&\delta\end{pmatrix},
 \qquad \det M=1.
\end{equation}
Choose based loops encircling $u=1$ and $u=-1$ once in the positive direction, ordered so that their product is a counterclockwise large circle in the $u$-plane. Because the local coordinate at infinity is $1/u$, this large circle is the inverse of a positive local loop about infinity. In the marking of the quartic family specified below, the monodromy matrices are
\begin{equation}\label{eq:monodromy}
 M_{+1}=\begin{pmatrix}1&0\\-1&1\end{pmatrix},\quad
 M_{-1}=\begin{pmatrix}-1&4\\-1&3\end{pmatrix},\quad
 M_\infty=M_{+1}M_{-1}=\begin{pmatrix}-1&4\\0&-1\end{pmatrix}.
\end{equation}
These matrices generate the integral monodromy image
\begin{equation}\label{eq:monodromyimage}
 G_{\mathrm{mon}}=\langle M_{+1},M_{-1}\rangle.
\end{equation}
Although $G_{\mathrm{mon}}$ does not contain $-I$, adjoining this central element gives
\begin{equation}\label{eq:groups}
 \langle G_{\mathrm{mon}},-I\rangle
 =\Gamma^0(4)
 =\left\{\begin{pmatrix}\alpha&\beta\\\gamma&\delta\end{pmatrix}
 \in\mathrm{SL}_2(\ZZ):\ \beta\equiv0\pmod4\right\}.
\end{equation}
Thus $G_{\mathrm{mon}}$ has index two in $\Gamma^0(4)$, and both groups have the same image in $\mathrm{PSL}_2(\ZZ)$. They must nevertheless be distinguished for the integral local system, even though $-I$ acts trivially on period ratios and modular forms of even weight. To verify the group assertions, use $t=\tau/2$, whose Legendre interpretation is established in Section~\ref{sec:schwarzian}. In this variable, $M_{+1}$ and $M_\infty$ act through
\begin{equation*}
 A=\begin{pmatrix}1&0\\-2&1\end{pmatrix},\qquad
 -B=-\begin{pmatrix}1&-2\\0&1\end{pmatrix}.
\end{equation*}
Here $\Gamma(2)=\{M\in\mathrm{SL}_2(\ZZ):M\equiv I\pmod2\}$. The Euclidean algorithm with even row operations shows that $A$, $B$, and $-I$ generate $\Gamma(2)$. Their projective images $\bar A$ and $\bar B$ freely generate $P\Gamma(2)$, the image of $\Gamma(2)$ in $\mathrm{PSL}_2(\ZZ)$. Indeed, on $\RR\cup\{\infty\}$ each nonzero power of $B$ sends $|t|<1$ into $|t|>1$, while each nonzero power of $A$ sends $|t|>1$ into $|t|<1$, so the ping-pong argument applies. Since $M_{+1}$ and $M_\infty$ generate $G_{\mathrm{mon}}$, this freeness excludes a word representing $-I$. Returning to $\tau$ proves \eqref{eq:groups} and the index-two assertion.

\subsection{The Picard--Fuchs equation and the prepotential} \label{sec:pfandprep}

We realize the special coordinates as periods of a meromorphic differential \cite{SWTheory,Lerche} on the quartic family
\begin{equation}\label{eq:quartic}
 \mathscr E_u:\quad y^2=(x^2-u)^2-1,\qquad u\in B.
\end{equation}
We write $\mathscr E_u$ for the smooth projective completion, with one of its two points at infinity chosen as the origin. The Seiberg--Witten differential is
\begin{equation}\label{eq:SWdifferential}
 \widehat\lambda_{\mathrm{SW}}=\frac{i}{\sqrt2\pi}\frac{x^2dx}{y}.
\end{equation}
The poles of this differential at infinity are of the second kind: their residues are zero. Its phase makes the electric period over a primitive cycle positive for real $u>1$. We orient the cycle around the cut $[\sqrt{u-1},\sqrt{u+1}]$ so that it traverses the branch $y=+i\sqrt{1-(x^2-u)^2}$ from left to right, giving
\begin{equation}\label{eq:electricintegral}
 a(u)=\frac{1}{\sqrt2\pi}\int_0^\pi\sqrt{u+\cos\theta}\,d\theta
 \sim\sqrt{\frac u2}.
\end{equation}
We orient the primitive magnetic vanishing cycle so that $a_D(u)\sim i(u-1)/2$ as $u\to1$. With the phases and signs of both periods fixed, we continue all subsequent formulas from this real branch.

Differentiating at fixed $x$ gives the identity
\begin{equation*}
 \partial_u\!\left(\frac{x^2dx}{y}\right)
  =\frac{dx}{2y}-\frac12d_x\!\left(\frac{x}{y}\right),
\end{equation*}
where $d_x$ denotes differentiation along the elliptic fiber. Integrating over a closed cycle removes the exact term and gives
\begin{equation}\label{eq:gm}
 \nabla_{\partial_u}[\widehat\lambda_{\mathrm{SW}}]
 =\frac{i}{2\sqrt2\pi}\left[\frac{dx}{y}\right]
 \quad\text{in }H^1_{\mathrm{dR}}(\mathscr E_u).
\end{equation}
Here $[\,\cdot\,]$ denotes a de Rham cohomology class and $\nabla$ is the Gauss--Manin connection, which differentiates cohomology classes as the parameter varies. In particular, the derivatives $\partial_u a$ and $\partial_u a_D$ are periods of a holomorphic differential.

Every period $\varpi(u)$ of $\widehat\lambda_{\mathrm{SW}}$ satisfies the Picard--Fuchs equation
\begin{equation}\label{eq:PFu}
 \frac{d^2\varpi}{du^2}+\frac{1}{4(u^2-1)}\varpi=0.
\end{equation}
For a direct verification, put $D=u^2-1$ and $\eta=x^2dx/y$. Differentiation on the affine curve gives
\begin{equation}\label{eq:PFexact}
 \left(\partial_u^2+\frac1{4D}\right)\eta
 =d_x\!\left[\frac{x^3(x^4-4ux^2+3u^2+1)}{4D\,y^3}\right].
\end{equation}
The function in brackets is meromorphic on the smooth projective fiber, so its differential has zero period. Integrating the identity proves \eqref{eq:PFu}.

Where $u\ne0$, the change of variable $\varphi=u^{-2}$ transforms the equation into
\begin{equation}\label{eq:PFvarphi}
 \left[\left(\theta_\varphi+\frac14\right)^2
 -\varphi\theta_\varphi\left(\theta_\varphi+\frac12\right)\right]\varpi=0,
 \qquad \theta_\varphi=\varphi\frac{d}{d\varphi},
\end{equation}
or, with primes denoting $\varphi$ derivatives,
\begin{equation}\label{eq:PFstandard}
\varpi''+\frac{3}{2\varphi}\varpi' +\frac{1}{16\varphi^2(1-\varphi)}\varpi=0.
\end{equation}

Choose a simply connected sector of $0<|\varphi|<1$ and continue the branches of $\varphi^{-1/4}$ and $\log\varphi$ from $0<\varphi<1$. A Frobenius basis for the Picard--Fuchs equation is then
\begin{align}
 \varpi_0&=\varphi^{-1/4}\sum_{n=0}^\infty c_n\varphi^n,\qquad c_n=\frac{(1/4)_n(-1/4)_n}{(n!)^2},\label{eq:w0}\\
 \varpi_1&=\varpi_0\log\varphi+
 \varphi^{-1/4}\sum_{n=1}^\infty d_n\varphi^n,\label{eq:w1}
\end{align}
with
\begin{align}\label{eq:dn}
 d_n=c_n\big[&2\psi_0(1)-2\psi_0(n+1)
 +\psi_0(n+1/4)-\psi_0(1/4)\notag\\
 &+\psi_0(n-1/4)-\psi_0(-1/4)\big].
\end{align}
Here $(b)_n=\Gamma(b+n)/\Gamma(b)$ is the Pochhammer symbol, with $(b)_0=1$, and $\psi_0(z)=\Gamma'(z)/\Gamma(z)$ is the digamma function. The symbol ${}_pF_q$ denotes the generalized hypergeometric function. In particular,
\begin{equation}
 \varpi_0=\varphi^{-1/4}\,{}_2F_1(-1/4,1/4;1;\varphi).
\end{equation}
The electric-period normalization and the magnetic vanishing condition give
\begin{equation}\label{eq:periodnorm}
 a=\frac{\varpi_0}{\sqrt2},\qquad
 a_D=\frac{\varpi_1+(4-6\log2)\varpi_0}{\sqrt2\,\pi i}.
\end{equation}
Equivalently, the magnetic period is
\begin{equation}\label{eq:magneticintegral}
 a_D(u)=\frac{i}{2}(u-1)\,{}_2F_1\!\left(\frac12,\frac12;2;\frac{1-u}{2}\right).
\end{equation}
Their limits at the singular fiber are
\begin{equation*}
 a(1)=\frac2\pi,\qquad a_D(1)=0,\qquad \partial_u a_D(1)=\frac i2,
 \qquad \frac{\varpi_1(1)}{\varpi_0(1)}=6\log2-4.
\end{equation*}
These limits are taken along the chosen branch. At weak coupling, as $u\to+\infty$, the periods give
\begin{align}
 a&=\sqrt{\frac u2}\left(1-\frac1{16u^2}+O(u^{-4})\right),\label{eq:aasy}\\
 \tau&=\frac{i}{\pi}\left(2\log u+6\log2+O(u^{-2})\right).\label{eq:tauasy}
\end{align}
The following Wronskian identities use the convention $W(f,g)=fg'-f'g$, with derivatives taken in the indicated coordinate.
\begin{lemma}[Wronskians]\label{lem:wronskian}
With derivatives in $u$,
\begin{equation}\label{eq:Wu}
 a\,\partial_u a_D-(\partial_u a)a_D=\frac{i}{\pi},
 \qquad
 \partial_u\tau=\frac{i}{4\pi(u^2-1)(\partial_u a)^2}.
\end{equation}
With primes denoting $\varphi$-derivatives,
\begin{equation}\label{eq:Wvarphi}
 W(\varpi_0,\varpi_1)=\varphi^{-3/2},\qquad
 \tau'=\frac{1}{16\pi i\,\varphi^{7/2}(1-\varphi)(\varpi_0')^2}.
\end{equation}
\end{lemma}
\begin{proof}
Because \eqref{eq:PFu} has no first-derivative term, its Wronskian is constant; the asymptotics of \eqref{eq:periodnorm} fix its value at $i/\pi$. Differentiating \eqref{eq:defntau} and substituting \eqref{eq:PFu} gives the formula for $\partial_u\tau$.

For the Frobenius basis, Abel's identity applied to \eqref{eq:PFstandard} gives $W(\varpi_0,\varpi_1)=c\varphi^{-3/2}$, and the leading terms fix $c=1$. The chain rule converts the $u$-coordinate identity into the formula for $\tau'$.
\end{proof}

With the same normalization, the weak-coupling prepotential has the expansion \cite[Section~2.6]{Lerche}
\begin{equation}\label{eq:F}
 \mathcal F(a)=\frac{ia^2}{2\pi}
 \left(2\log a^2-6+8\log2-
 \sum_{\ell\geq1}f_\ell a^{-4\ell}\right),
\end{equation}
where the logarithmic and quadratic terms are perturbative and the coefficients $f_\ell$ multiply instanton corrections. The first four coefficients are
\begin{equation}
 f_1=\frac1{2^5},\quad f_2=\frac5{2^{14}},\quad
  f_3=\frac3{2^{18}},\quad f_4=\frac{1469}{2^{31}}.
\end{equation}
The expansion converges in a sufficiently small weak-coupling sector. To see this, write $a=\varphi^{-1/4}h(\varphi)/\sqrt2$, where $h$ is holomorphic and $h(0)=1$. The map $a^{-4}=4\varphi h(\varphi)^{-4}$ has a holomorphic inverse near zero, so the Frobenius formulas express $a_D/a$ as the prescribed logarithmic term plus a holomorphic function of $a^{-4}$. Integrating $a_D$ with respect to $a$ therefore gives \eqref{eq:F} as a convergent expansion, up to an additive constant.

\section{The Schwarzian derivative and modularity}\label{sec:schwarzian}

The Seiberg--Witten period ratio and the effective coupling define two projective connections whose difference is a quadratic differential. The Picard--Fuchs equations determine this difference explicitly, while a general identity for rank-two flat symplectic bundles gives the formulation needed later for the leading logarithmic Yukawa tensor. Legendre uniformization expresses the coefficient as a modular form of weight four and gives its Fourier expansion; a degree-two isogeny accounts for the coupling normalization.

\subsection{The Schwarzian derivative}
Let $f$ be holomorphic and locally univalent on a coordinate domain, so that $f'(z)\ne0$. Its Schwarzian derivative is
\begin{equation}\label{eq:schwarzdef}
 \Sch{f}{z}=\frac{f'''}{f'}-\frac32\left(\frac{f''}{f'}\right)^2.
\end{equation}
For a locally univalent map to $\PP^1$, we may use any projective target coordinate, since postcomposition by a M\"obius transformation leaves the Schwarzian unchanged. In particular,
\begin{equation*}
 \Sch{(Az+B)/(Cz+D)}{z}=0,\qquad AD-BC\ne0.
\end{equation*}
Conversely, the Schwarzian of a locally univalent function vanishes identically if and only if the function is locally a M\"obius transformation. Local univalence is essential here, since a constant function has no Schwarzian.

Under a locally invertible change of source coordinate $w=w(z)$, the Schwarzian obeys the chain rule
\begin{equation}\label{eq:chain}
 \Sch{f}{z}=\Sch{f}{w}\left(\frac{dw}{dz}\right)^2+\Sch{w}{z}.
\end{equation}
Two locally univalent functions on a connected domain have the same Schwarzian if and only if they differ by a constant M\"obius transformation of the target \cite{Hille,Thurston}.

If $f_1$ and $f_2$ are linearly independent solutions of
\[
f''+P f'+Qf=0,
\]
where primes denote $z$ derivatives and $P,Q$ are holomorphic, then their quotient satisfies
\begin{equation}\label{eq:ratioidentity}
 \Sch{f_1/f_2}{z}=2Q-P'-\tfrac12P^2.
\end{equation}
The substitution $f=\exp(-\frac12\int P\,dz)h$ removes the first-derivative term and gives
\[
 h''+(Q-P'/2-P^2/4)h=0
\]
without changing the quotient of solutions. Differentiating that quotient gives \eqref{eq:ratioidentity}. For Schwarzian equations associated with automorphic forms, see \cite{Automorphic}.

\begin{proposition}[Rational Schwarzian identities]\label{prop:rational}
On $B$, the period ratios satisfy
\begin{equation}\label{eq:rationals}
 \Sch{a_D/a}{u}=\frac{1}{2(u^2-1)},\qquad
 \Sch{\tau}{u}=\frac{u^2+3}{2(u^2-1)^2}.
\end{equation}
Consequently,
\begin{equation}\label{eq:rationalcomparison}
 \Sch{a_D/a}{\tau}
 =-\frac{2}{(u^2-1)^2}\left(\frac{du}{d\tau}\right)^2
 =-\frac{1}{2\varphi(1-\varphi)^2}
 \left(\frac{d\varphi}{d\tau}\right)^2.
\end{equation}
\end{proposition}
\begin{proof}
The first identity follows from \eqref{eq:PFu} and \eqref{eq:ratioidentity}. Differentiating \eqref{eq:PFu} and eliminating $\varpi$ gives the equation for $v=\partial_u\varpi$:
\begin{equation}\label{eq:derivativePF}
 v''+\frac{2u}{u^2-1}v'+\frac{1}{4(u^2-1)}v=0,
 \qquad '=\frac d{du}.
\end{equation}
By the Wronskian identity \eqref{eq:Wu} and the equation \eqref{eq:PFu}, the functions $\partial_u a_D$ and $\partial_u a$ are independent on $B$. Applying \eqref{eq:ratioidentity} to their quotient proves the second formula in \eqref{eq:rationals}. Subtracting the formulas and using the chain rule gives \eqref{eq:rationalcomparison}. At a pole of an affine period ratio, we evaluate the identity in a projective target chart.
\end{proof}

The expression in $\varphi$ is used only where $u\ne0$; the quadratic differential extends across $u=0$ in the local coordinate $u$.

\subsection{A Schwarzian identity for rank-two flat symplectic bundles}
The classical relation between second-order linear differential equations and the Schwarzian derivative \cite{Hille,Automorphic} has a rank-two formulation that compares the ratio of periods with the ratio of their derivatives. We work in the flat symplectic period formalism of special K\"ahler geometry and its Picard--Fuchs systems \cite{CeresoleSpecialGeometry,FreedSpecialKahler}, and prove the identity needed for the boundary construction.

Let $\mathbb V$ be a holomorphic vector bundle of rank two on a complex curve, equipped with a flat connection $\nabla$ and a nondegenerate alternating pairing $\langle\ ,\ \rangle$ preserved by $\nabla$. Choose a holomorphic section $s$. In a local flat symplectic frame, write $s=(s_0,s_1)$ and use the convention
\begin{equation}
 \langle(a,b),(c,d)\rangle=ad-bc.
\end{equation}
With $\nabla_z=\nabla_{\partial/\partial z}$, define the one-form $W$ and cubic differential $\Ygen$ by
\begin{equation}\label{g:WY}
 W=\langle s,\nabla s\rangle,
 \qquad
 \Ygen=\langle\nabla_zs,\nabla_z^2s\rangle(dz)^3.
\end{equation}
Under a change of coordinate, the extra chain-rule term in $\nabla_z^2s$ is proportional to $\nabla_zs$ and has zero pairing with it. Thus $\Ygen$ is a well-defined cubic differential. On the relevant affine charts, set $r=s_1/s_0$ and $\tau=(\partial_zs_1)/(\partial_zs_0)$. Taking $s=(a,a_D)$ recovers the Seiberg--Witten ratios in \eqref{eq:basic}.

For a nowhere-vanishing holomorphic one-form $\alpha$, write
\begin{equation}
 P_\alpha(z)=\Sch{\int\alpha}{z}.
\end{equation}
On a simply connected coordinate neighborhood, the primitive is unique up to an additive constant. Its Schwarzian has the transformation law \eqref{eq:chain}, which defines a projective connection. The inhomogeneous terms cancel when two such coefficients are subtracted, leaving a quadratic differential.

\begin{theorem}[Intrinsic Schwarzian identity]\label{thm:ranktwo}
On a sufficiently small simply connected neighborhood where $W$ and $\Ygen$ are nowhere zero, the quotients $r$ and $\tau$ are locally univalent as maps to $\PP^1$. Evaluate their Schwarzians in projective target charts when necessary, and choose primitives
\begin{equation}\label{g:primitives}
 dx=W,\qquad dF=\frac{\Ygen}{W^2}.
\end{equation}
Then
\begin{equation}\label{g:rankidentity}
 \bigl(\Sch{r}{z}-\Sch{\tau}{z}\bigr)(dz)^2
 =\bigl(\Sch{x}{z}-\Sch{F}{z}\bigr)(dz)^2.
\end{equation}
Here $\Ygen/W^2$ is a holomorphic one-form, obtained by dividing the indicated tensor powers. The quadratic differential is independent of the flat symplectic frame and of constant nonzero rescalings of $s$ or of the pairing.
\end{theorem}
\begin{proof}
The nonvanishing of $W$ makes its primitive $x$ a local coordinate with $\langle s,s_x\rangle=1$. Differentiating this equality gives $\langle s,s_{xx}\rangle=0$, so $s_{xx}$ is proportional to $s$. Pairing with $s_x$ determines the proportionality factor:
\begin{equation}
 s_{xx}=-k s,\qquad k=\frac{\Ygen}{W^3}.
\end{equation}
Since $k=\Ygen/W^3$ is holomorphic and nowhere zero, \eqref{eq:ratioidentity} gives $\Sch{r}{x}=2k$. Differentiating the equation and eliminating $s$ yields an equation for the components of $v=s_x$:
\begin{equation}
 v_{xx}-\frac{k_x}{k}v_x+kv=0.
\end{equation}
Their quotient therefore has Schwarzian
\begin{equation}
 \Sch{\tau}{x}
 =2k+\frac{k_{xx}}k-\frac32\left(\frac{k_x}k\right)^2.
\end{equation}
Since $F_x=k$, subtraction gives $\Sch{r}{x}-\Sch{\tau}{x}=-\Sch{F}{x}$. The chain rule then proves \eqref{g:rankidentity} in any source coordinate $z$. Changing the flat symplectic frame applies a constant M\"obius transformation to each period ratio, while a nonzero constant rescaling of $s$ or of the pairing merely rescales $x$ and $F$ by constants. Their Schwarzians are unchanged in either case.
\end{proof}

\subsection{Modular forms and the Schwarzian derivative}
We obtain the modular expression by comparing the derivative-period equation \eqref{eq:derivativePF} with the Picard--Fuchs equation of the Legendre family
\begin{equation}
   \mathscr E^{\mathrm L}_\lambda:\quad y^2=x(x-1)(x-\lambda),\qquad \lambda\in\PP^1\setminus\{0,1,\infty\}.
\end{equation}
For $0<\lambda<1$, choose an oriented integral basis of $H_1(\mathscr E^{\mathrm L}_\lambda,\ZZ)$ for which the periods of $dx/y$ are
\begin{align}
    \omega_0 (\lambda) &= 2\pi\, {}_2F_1(1/2,1/2;1;\lambda),\\
    \omega_1 (\lambda) &= 2\pi i\, {}_2F_1(1/2,1/2;1;1-\lambda).
\end{align}
We continue these periods analytically from $0<\lambda<1$. Their factor $2\pi$ comes from integrating $dx/y$ over primitive closed cycles; using $dx/(2y)$ would give $\pi$. Both periods satisfy
\begin{equation}
    \lambda (1-\lambda) \omega'' +(1-2\lambda) \omega' -\frac{1}{4} \omega =0,
\end{equation}
where primes denote differentiation with respect to $\lambda$. Their ratio is
\begin{equation}
    t=\frac{\omega_1}{\omega_0}=i \frac{{}_2F_1(1/2,1/2;1;1-\lambda)}{{}_2F_1(1/2,1/2;1;\lambda)}.
\end{equation}
Under the change of variable
\begin{equation}\label{eq:legendrechange}
 \lambda=\frac{2}{u+1},\qquad
 \varphi=\frac{\lambda^2}{(2-\lambda)^2},
\end{equation}
direct substitution shows that $(u+1)^{-1/2}\omega(\lambda(u))$ satisfies \eqref{eq:derivativePF}. Thus the two solution bases are related by a constant invertible matrix, and their ratios $t$ and $\tau$ by a constant M\"obius transformation. We determine this transformation from the hypergeometric asymptotics as $u\to+\infty$:
\begin{equation*}
 t=\frac{i}{\pi}\log(8u)+o(1),\qquad
 \tau=\frac{2i}{\pi}\log(8u)+o(1).
\end{equation*}
Since both ratios tend to infinity, the transformation is affine; its slope and constant term follow by comparing the coefficients of $\log u$ and the remaining constants. We obtain
\begin{equation}\label{eq:halfcoupling}
 t=\frac\tau2.
\end{equation}
The factor two has a geometric explanation. The cubic family
\begin{equation*}
 \mathscr E^{\mathrm c}_u:\quad Y^2=(z^2-1)(z-u)
\end{equation*}
is transformed to Legendre form by $X=(z+1)/(u+1)$, together with a compatible rescaling of $Y$. Its Legendre parameter is $\lambda=2/(u+1)$. The map
\begin{equation}\label{b:isogeny}
 f:\mathscr E_u\longrightarrow\mathscr E^{\mathrm c}_u,
 \qquad z=u-x^2,\qquad Y=i x y,
 \qquad f^*(dz/Y)=2i\,dx/y
\end{equation}
extends to a degree-two morphism between the smooth projective curves, with fixed-point-free deck involution $(x,y)\mapsto(-x,-y)$ on each smooth fiber. Choosing compatible origins makes it an isogeny. The relation $t=\tau/2$ therefore compares two isogenous elliptic curves, rather than two integral symplectic bases on the same curve.

For $t\in\HH$, the Jacobi theta constants are defined by
\begin{equation}
 \theta_2(t)=\sum_{n\in\ZZ}e^{\pi it(n+1/2)^2},\qquad
 \theta_3(t)=\sum_{n\in\ZZ}e^{\pi itn^2},\qquad
 \theta_4(t)=\sum_{n\in\ZZ}(-1)^ne^{\pi itn^2}.
\end{equation}
The Legendre uniformization is \cite[Sections~19.5 and 23.15]{DLMF}
\begin{equation}\label{eq:legendre}
 \lambda=\frac{\theta_2(t)^4}{\theta_3(t)^4},
\end{equation}
and the identities needed below are
\begin{equation}\label{eq:thetaidentities}
 \frac{d\lambda}{dt}=\pi i\,\theta_3(t)^4\lambda(1-\lambda),\qquad
 E_4(t)-E_4(2t)=\frac{15}{16}\theta_2(t)^8.
\end{equation}
The derivative identity follows from the Wronskian of $\omega_0$ and $\omega_1$, together with
\[
 {}_2F_1(1/2,1/2;1;\lambda)=\theta_3(t)^2.
\]
The second follows from the theta expression $E_4=(\theta_2^8+\theta_3^8+\theta_4^8)/2$, the quartic identity $\theta_3^4=\theta_2^4+\theta_4^4$, and the duplication identities for theta constants \cite{DLMF}.

The identity below expresses the Schwarzian as a weight-four modular form in $t=\tau/2$ for $\Gamma_0(2)$, the subgroup of $\mathrm{SL}_2(\ZZ)$ with even lower-left entry.

\begin{theorem}[Modular Schwarzian identity]\label{thm:modular}
With the conventions \eqref{eq:coordinates} and \eqref{eq:periodnorm},
\begin{equation}\label{eq:mainmodular}
 \Sch{a_D/a}{\tau}=\frac{\pi^2}{8}\theta_2(\tau/2)^8
 =32\pi^2 E_{\infty,4}(\tau/2),
\end{equation}
where
\begin{equation}\label{eq:Edef}
 E_4(t)=1+240\sum_{n\geq1}\sigma_3(n)e^{2\pi int},\qquad
 E_{\infty,4}(t)=\frac{E_4(t)-E_4(2t)}{240}.
\end{equation}
Here $\sigma_r(n)=\sum_{d\mid n}d^r$ is the divisor sum. The series $E_4$ is normalized to have constant term one.
\end{theorem}
\begin{proof}
Using $u=2/\lambda-1$ and $t=\tau/2$ in \eqref{eq:thetaidentities} gives
\begin{equation}
 \frac{\partial_\tau u}{u^2-1}
 =-\frac{\pi i}{4}\theta_2(t)^4.
\end{equation}
Substituting into \eqref{eq:rationalcomparison} gives the theta expression in \eqref{eq:mainmodular}, and the second identity in \eqref{eq:thetaidentities} converts it to the Eisenstein-series expression. The result continues analytically from the chosen branch.
\end{proof}

\begin{remark}[An alternative coupling coordinate]\label{rem:theta4}
A change of coupling coordinate gives an expression in $\theta_4$. Put $\sigma=-2/\tau=-1/t$; because $\tau=-2/\sigma$ is fractional-linear, the chain rule \eqref{eq:chain} and theta inversion give
\begin{equation}
 \Sch{r}{\sigma}=\Sch{r}{\tau}\frac{4}{\sigma^4}
 =\frac{\pi^2}{2}\theta_4(\sigma)^8.
\end{equation}
We have used $\theta_2(-1/\sigma)^8=\sigma^4\theta_4(\sigma)^8$. The rescaling in $\sigma=-2/\tau$ means that this is not an integral determinant-one electric-magnetic frame change; it sends the weak-coupling limit $\tau\to i\infty$ to $\sigma\to0$, rather than $\sigma\to i\infty$.
\end{remark}

Expanding the Eisenstein series in \eqref{eq:Edef} gives
\begin{equation}\label{eq:Eseries}
 E_{\infty,4}(t)=q+8q^2+28q^3+64q^4+126q^5
 +224q^6+344q^7+512q^8+O(q^9).
\end{equation}
Its $n$th coefficient is $\sigma_3(n)-\sigma_3(n/2)$, with the second term interpreted as zero for odd $n$. The normalized theta series also gives
\begin{equation}\label{eq:triangular}
 E_{\infty,4}(t)=q\left(\sum_{k\geq0}q^{k(k+1)/2}\right)^8.
\end{equation}
Thus the coefficient of $q^n$ is the number of ordered representations of $n-1$ as a sum of eight nonnegative triangular numbers.

Although $E_{\infty,4}$ vanishes at infinity, it is not cuspidal. At the cusp $0$, the transformation law of $E_4$ gives
\begin{equation}
 t^{-4}E_{\infty,4}(-1/t)
 =\frac{E_4(t)-2^{-4}E_4(t/2)}{240}
 \longrightarrow\frac1{256}\quad(t\to i\infty).
\end{equation}

For reference, inversion of the period map gives
\begin{align}
 \varphi&=64q-2560q^2+84736q^3-2551808q^4+O(q^5),\label{eq:varphiseries}\\
 \frac{64}{\varphi}&=q^{-1}+40+276q-2048q^2+11202q^3-49152q^4+O(q^5).
 \label{eq:hauptseries}
\end{align}
A Hauptmodul, that is, a meromorphic coordinate on the compactified genus-zero modular curve, is
\begin{equation}\label{eq:hauptmodul}
 J_2(t):=\frac{64}{\varphi(t)}
 =\left(\frac{\eta(t)}{\eta(2t)}\right)^{24}+64,
 \qquad \eta(t)=q^{1/24}\prod_{n\geq1}(1-q^n).
\end{equation}
Here $\eta$ denotes the Dedekind eta function. The equality follows from the theta product formulas:
\[
 \left(\frac{\eta(t)}{\eta(2t)}\right)^{24}
 =\frac{256(1-\lambda)}{\lambda^2}
 =\frac{64(2-\lambda)^2}{\lambda^2}-64.
\]
Under $t\mapsto t+1$, we have $\lambda\mapsto\lambda/(\lambda-1)$ and $u\mapsto-u$, so $J_2$ is invariant. The rational function $64(2-\lambda)^2/\lambda^2$ has degree two on $X(2)$ and descends to a degree-one function on $X_0(2)$, where $X(2)$ and $X_0(2)$ are the compactified modular curves of $\Gamma(2)$ and $\Gamma_0(2)$. Thus $J_2$ identifies $X_0(2)$ with $\PP^1$.

Set $S(\tau)=\Sch{r}{\tau}$. Since monodromy applies the same M\"obius transformation to $r$ and $\tau$, target invariance and the source chain rule give the weight-four transformation law for $G_{\mathrm{mon}}$. The central element $-I$ acts trivially in even weight, so the same law holds for the central enlargement:
\begin{equation}\label{eq:weight4}
 S(M\tau)=(\gamma\tau+\delta)^4S(\tau),\qquad M\in\Gamma^0(4).
\end{equation}
The weight-four form $S$ is the coefficient of the invariant quadratic differential $\Quad=S(\tau)(d\tau)^2$ in the coupling coordinate. Its Fourier expansion defines a holomorphic germ at $q=0$, which we also denote by $S(q)$.

The perturbative prepotential consists of the logarithmic and quadratic terms of \eqref{eq:F} \cite{Seiberg,SWTheory}. Differentiating gives
\begin{equation}\label{eq:perturbative}
 r_{\mathrm{pert}}=\tau_{\mathrm{pert}}-\frac{4i}{\pi},\qquad
 \Sch{r_{\mathrm{pert}}}{\tau_{\mathrm{pert}}}=0.
\end{equation}
This comparison is made in the chosen weak-coupling frame, without assuming that global duality preserves the perturbative truncation. In the next section, we prove the converse for a prescribed class of normalized prepotentials and recover their instanton coefficients from the full Schwarzian germ.

\section{The Schwarzian derivative and instanton corrections}\label{sec:quantum}
For a class of convergent local prepotentials with fixed weak-coupling normalization, the holomorphic Schwarzian germ determines the instanton corrections. We prove that it vanishes identically if and only if all instanton corrections vanish, and that the full germ determines the prepotential up to an additive constant. We then derive a triangular, nonlinear relation between Schwarzian Fourier coefficients and gauge-instanton coefficients.

\subsection{The perturbative prepotential and the vanishing criterion}
The holomorphic prepotential determines the two-derivative part of the low-energy action. In pure $N=2$ gauge theory its perturbative contribution is one-loop exact, while gauge instantons contribute inverse powers of the electric special coordinate \cite[Section~2.3]{SWTheory}. Instanton counting provides another way to compute these corrections \cite{Nekrasov}.

Fix the period normalization of Section~\ref{sec:sw}, set the dynamical scale to one, and choose a logarithm branch near weak coupling. We consider prepotentials of the form
\begin{equation}\label{q:Fclass}
 \mathcal F(a)=\frac{ia^2}{2\pi}
 \left[2\log a^2-6+8\log2-\Phi(a^{-4})\right]+C,
 \qquad \Phi(z)=\sum_{n\ge1}f_nz^n\in z\CC\{z\}.
\end{equation}
Here $\CC\{z\}$ is the ring of power series convergent near $z=0$, and $z\CC\{z\}$ consists of those with zero constant term. On the chosen branch, $\log a^2=2\log a$. The constant $C$ does not affect the periods; the fixed quadratic normalization and integral cycle marking exclude linear or quadratic additions to $\mathcal F$. Set
\begin{equation}\label{q:ratios}
 r=\frac{\mathcal F'(a)}a,\qquad \tau=\mathcal F''(a),\qquad
 b=\frac{4i}{\pi},\qquad q=e^{\pi i\tau},\qquad S=\Sch{r}{\tau}.
\end{equation}
Differentiation gives
\begin{align}\label{q:rseries}
 r&=b[\log(4a)-1]+\sum_{n\ge1}\gamma_na^{-4n},
 &\gamma_n&=\frac{i}{\pi}(2n-1)f_n,\notag\\
 \tau&=b\log(4a)+\sum_{n\ge1}(1-4n)\gamma_na^{-4n},
 &q&=\frac{a^{-4}}{256}\bigl(1+O(a^{-4})\bigr).
\end{align}
We use $\gamma_n$ for the coefficients in these period expansions, keeping them distinct from the Frobenius coefficients $c_n$ in \eqref{eq:w0}. The relation $q=a^{-4}/256+O(a^{-8})$ is locally holomorphically invertible. On the chosen branch, $a\sim q^{-1/4}/4$ and $r=\tau-b+O(q)$. Near the cusp, $\tau$ is locally invertible as a function of $a$, and $r_\tau=1+O(q)\ne0$ ensures that $\Sch{r}{\tau}$ is well defined.

A homogeneous quadratic prepotential makes both $r$ and $\tau$ constant, so the Schwarzian is undefined. Our reference is therefore the one-loop perturbative prepotential, for which $r_{\mathrm{pert}}=\tau_{\mathrm{pert}}-b$ and the Schwarzian vanishes.

Define the deviation from the perturbative relation by
\begin{equation}\label{q:delta}
 \delta(\tau)=r-\tau+b
 =\frac{4i}{\pi}\sum_{n\ge1}n(2n-1)f_n a(\tau)^{-4n}.
\end{equation}
The Schwarzian is then given exactly by
\begin{equation}\label{q:Sdelta}
 S=\frac{\delta'''}{1+\delta'}
 -\frac32\left(\frac{\delta''}{1+\delta'}\right)^2,
 \qquad '\,=\frac d{d\tau}.
\end{equation}
The nonlinear dependence of $S$ on $\delta$ is explicit in \eqref{q:Sdelta}. Since $\partial_\tau=\pi i q\partial_q$ and $\delta=O(q)$, we have $S\in q\CC\{q\}$. The Schwarzian leaves a M\"obius ambiguity in $r$; the prescribed weak-coupling asymptotics remove it in the proofs below.

\begin{theorem}[Vanishing criterion]\label{q:vanishing}
For the normalized convergent prepotentials \eqref{q:Fclass},
\begin{equation}\label{q:iff}
 S\equiv0\quad\Longleftrightarrow\quad
 f_n=0\text{ for every }n\ge1.
\end{equation}
The hypothesis is that $S$ vanishes identically; vanishing at an isolated point is insufficient.
\end{theorem}
\begin{proof}
If $S\equiv0$, then $r$ is a constant M\"obius transformation of $\tau$. The asymptotic relation $r=\tau-b+o(1)$ as $\tau\to i\infty$ in the chosen weak-coupling sector rules out a nonconstant denominator, which would give a finite limit at infinity. The transformation is therefore affine, with slope one and constant term $-b$, so $r=\tau-b$ identically.

Differentiating $a_D=ar$ gives $\tau=r+a\,r_a$, where $r_a=\partial r/\partial a$, and hence $a\,r_a=b$. Substituting \eqref{q:rseries} forces every $\gamma_n$, and therefore every $f_n$, to vanish. Conversely, when $\Phi=0$ the perturbative affine relation holds and gives $S=0$.
\end{proof}

This equivalence uses both the fixed normalization and the weak-coupling asymptotics, and applies only to the prescribed class.

\subsection{Reconstruction of the prepotential and its coefficients}
\begin{theorem}[Reconstruction of the normalized prepotential]\label{q:reconstruct}
Within the normalized class \eqref{q:Fclass}, the holomorphic Schwarzian germ $S(q)$ determines the prepotential up to the additive constant $C$.
\end{theorem}
\begin{proof}
Suppose two prepotentials in \eqref{q:Fclass} have the same Schwarzian germ. Their ratios $r_1$ and $r_2$, expressed in the common coupling variable $\tau$, differ by a constant M\"obius transformation. Since both satisfy $r_j=\tau-b+O(q)$, we have $r_1\to\infty$ and $r_2-r_1\to0$ at weak coupling. These limits force the transformation to have constant denominator, slope one, and zero constant term, so it is the identity. Thus the germ $S$ and the prescribed asymptotics determine $r$.

Differentiating $a_D=ar$ and using $da_D=\tau\,da$ gives
\begin{equation}\label{q:aODE}
 \frac{d\log a}{d\tau}=\frac{r_\tau}{\tau-r}.
\end{equation}
Near the cusp, $\tau-r=b+O(q)$ is nonzero, so \eqref{q:aODE} determines $a$ up to a multiplicative constant. The asymptotic condition $a\sim q^{-1/4}/4$ fixes this constant. Recovering $a_D=ar$ and integrating $\mathcal F'(a)=a_D$ then determine $\mathcal F$ up to an additive constant.
\end{proof}

The reconstruction theorem establishes local uniqueness within the normalized class. It does not assert that an arbitrary holomorphic function $S(q)$ extends to a global Seiberg--Witten period system. For the actual periods, Theorem~\ref{thm:modular} supplies the germ.

To identify the first nonzero instanton coefficient, let $N$ be the least index with $f_N\ne0$ and use the local variable $x=\log a$, in which $\tau=r+r_x$. The Schwarzian chain rule gives
\begin{equation}\label{q:leading}
 S=4\pi^2N^4(2N-1)f_N a^{-4N}
       +O(a^{-4N-4}).
\end{equation}
The leading coefficient of the difference of the two $x$-Schwarzians is $-256N^4\gamma_N/b$. Dividing by $\tau_x^2=b^2+O(a^{-4N})$ gives \eqref{q:leading}. The prefactor multiplying the instanton coefficient is nonzero for every positive integer $N$.

At higher orders, expand in $a^{-4}$ and invert $q=a^{-4}/256+O(a^{-8})$ to obtain the triangular relation
\begin{equation}\label{q:triangular}
 [q^n]S=4\pi^2\,256^n n^4(2n-1)f_n
            +\mathcal P_n(f_1,\ldots,f_{n-1}).
\end{equation}
The polynomial $\mathcal P_n$ contains only lower-order coefficients. To see this, assign weight $j$ to $f_j$: the coefficient at order $n$ is homogeneous of total weight $n$, so nonlinear terms can involve only $f_1,\ldots,f_{n-1}$. The term linear in $f_n$ is given by \eqref{q:leading}, multiplied by $256^n$. Its coefficient is nonzero for every $n\ge1$, which allows the $f_n$ to be recovered recursively. The Fourier coefficients of $S$ are therefore nonlinear combinations of gauge-instanton coefficients, not counts of instantons.

For $\widehat q=256q$, direct expansion yields
\begin{align}\label{q:firstcoefficients}
 \frac{S}{\pi^2}={}&4f_1 \widehat q+(192f_2-56f_1^2)\widehat q^2\notag\\
 &+(1620f_3-4044f_1f_2+658f_1^3)\widehat q^3+O(\widehat q^4).
\end{align}
Substituting the modular identity \eqref{eq:mainmodular},
\begin{equation}
 S(q)=32\pi^2(q+8q^2+28q^3+\cdots),
\end{equation}
into these relations recovers, successively,
\begin{equation}
 f_1=\frac1{32},\qquad f_2=\frac5{2^{14}},\qquad
 f_3=\frac3{2^{18}}.
\end{equation}
More generally, \eqref{q:leading} shows that the first nonzero Schwarzian coefficient identifies the first nonzero instanton sector and determines its coefficient.

The reconstructed $\mathcal F$ determines the effective coupling, the special K\"ahler metric \eqref{eq:metric}, and the central charges $n_ea+n_ma_D$. It does not determine which charges occur as stable Bogomolny--Prasad--Sommerfield (BPS) particles \cite[Sections~2.3 and 4]{SWTheory}. The reconstruction requires the full holomorphic Schwarzian germ and its fixed weak-coupling normalization. Section~\ref{sec:comparison} recovers the associated quadratic differential from the leading logarithmic Yukawa tensor using the prescribed connection.

\section{The two-parameter mirror family and its rigid limit}\label{sec:mirror}
A blowup chart separates the tangential and normal directions in the rigid limit of the two-parameter mirror family. With the holomorphic three-form normalized, the full Yukawa tensor has explicit rational components in this chart. The leading tangential coefficient is proportional to the Seiberg--Witten cubic, as shown in Section~\ref{sec:comparison}; the mixed components also enter the Schwarzian reconstruction there.

\subsection{Family, discriminants, and Hodge normalization}
The two-parameter mirror family is constructed from the weighted homogeneous polynomial
\begin{equation}\label{eq:CY}
 p=z_1^{12}+z_2^{12}+z_3^6+z_4^6+z_5^2
 -12\psi z_1z_2z_3z_4z_5-2\phi z_1^6z_2^6
\end{equation}
of degree $12$ in $\PP^4_{1,1,2,2,6}$. The homogeneous coordinates $z_1,\ldots,z_5$ have weights $1,1,2,2,6$, and $\psi,\phi$ are the family parameters. To obtain a mirror fiber, we quotient the hypersurface $p=0$ by the diagonal phase symmetries that preserve its holomorphic three-form and take a crepant resolution \cite{Candelas,Kachru}. The parameter $\phi$ here is unrelated to the adjoint scalar of Section~\ref{sec:sw}.

We use the algebraic moduli coordinates
\begin{equation}\label{eq:xy}
 x=-\frac{\phi}{864\psi^6},\qquad y=\phi^{-2}.
\end{equation}
Here $x,y$ parametrize the mirror family, rather than an elliptic fiber. The large-complex-structure limit lies at $x=y=0$; the point-particle limit approaches $x=1,y=0$.

In these two coordinate systems, the conifold discriminant is given by
\begin{equation}\label{eq:deltas}
 \Delta_\psi=(864\psi^6+\phi)^2-1,\qquad
 \Delta_{xy}=(1-x)^2-x^2y.
\end{equation}
On the common coordinate domain, the two functions have the same zero locus and are related by
\begin{equation}\label{eq:deltafactor}
 \Delta_\psi=\frac{\phi^2}{x^2}\Delta_{xy}.
\end{equation}

Let $\Omega$ be a local nonvanishing section of the Hodge line of holomorphic three-forms on the smooth mirror fibers. Such a section is called a \emph{Hodge frame}. In a local flat cohomology frame, the Yukawa components are
\begin{equation}
 C_{ijk}=-\int\Omega\wedge\partial_i\partial_j\partial_k\Omega.
\end{equation}
The integral is over the Calabi--Yau fiber, and $\partial_i$ differentiates in a moduli direction using the Gauss--Manin connection. By Griffiths transversality and polarization orthogonality, $\int\Omega\wedge\nabla^m\Omega=0$ for $m<3$. These vanishings ensure that the components transform tensorially under changes of moduli coordinates and acquire a factor $h^2$ under $\Omega\mapsto h\Omega$; see \eqref{e:hodgeGaugeExact}.

We use the algebraic normalization of \cite[Section~8.1]{Candelas}, where it is called the $F_0=\varpi_0$ gauge. The period $\varpi_0$ in that reference is the Calabi--Yau fundamental period and is distinct from the elliptic period of \eqref{eq:w0}. With symmetric products as in the introduction, the cubic tensor is
\begin{equation}\label{eq:cubiccomponents}
 \mathcal C=C_{\psi\psi\psi}(d\psi)^3
 +3C_{\psi\psi\phi}(d\psi)^2d\phi
 +3C_{\psi\phi\phi}d\psi(d\phi)^2
 +C_{\phi\phi\phi}(d\phi)^3.
\end{equation}
The factors $1,3,3,1$ account for the symmetric products: evaluation on three copies of a tangent vector $v$ gives $\sum_{i,j,k}C_{ijk}v^iv^jv^k$. In the chosen Hodge frame, the components are
\begin{align}
 C_{\psi\psi\psi}&=-\frac{i}{64\pi^3}\frac{(12\psi)^9}{\Delta_\psi},\notag\\
 C_{\psi\psi\phi}&=-\frac{3i}{4\pi^3}\frac{(12\psi)^4}{\Delta_\psi},\notag\\
 C_{\psi\phi\phi}&=-\frac{i}{48\pi^3}
 \frac{(\phi+432\psi^6)(12\psi)^5}{(1-\phi^2)\Delta_\psi},\notag\\
 C_{\phi\phi\phi}&=-\frac{i}{6912\pi^3}
 \frac{(1+3\phi^2+1728\psi^6\phi)(12\psi)^6}
 {(1-\phi^2)^2\Delta_\psi}.\label{eq:CYcouplings}
\end{align}
The holomorphic three-form $\Omega$ in this normalization is a constant multiple of $\psi$ times the residue form \cite[Section~6.1]{Candelas}, with its Yukawa normalization fixed by \eqref{eq:CYcouplings}.

\subsection{The rigid limit and the Yukawa tensor}
Following \cite{Kachru}, introduce a tangential coordinate $u$ and a transverse coordinate $\epsilon$ near the rigid-limit divisor by
\begin{equation}\label{eq:blowup}
 \varphi=\frac{x^2y}{(x-1)^2}=u^{-2},\qquad
 x-1=\epsilon u.
\end{equation}
Equivalently,
\begin{equation}\label{eq:xyexplicit}
 x=1+\epsilon u,\qquad y=\frac{\epsilon^2}{(1+\epsilon u)^2}.
\end{equation}
Choose the branches so that
\begin{equation}\label{eq:phipsi}
 \phi=-\epsilon^{-1}-u,\qquad \psi^6=\frac{1}{864\epsilon}.
\end{equation}
In particular, $\psi$ is constant on each fixed-$\epsilon$ slice. The discriminants become
\begin{equation}\label{eq:blowupdiscriminants}
 \Delta_\psi=u^2-1,\qquad
 \Delta_{xy}=\epsilon^2(u^2-1)
 =\epsilon^2\frac{1-\varphi}{\varphi}.
\end{equation}
This chart comes from the double cover with coordinate $w=x\sqrt y$. Blowing up $(x-1,w)=(0,0)$ gives $w=\epsilon$, $x-1=\epsilon u$, so the exceptional divisor $\epsilon=0$ parametrizes the rigid limit. The strict transform of the conifold discriminant meets this divisor at $u=\pm1$; both points map to $\varphi=1$. The rigid-limit divisor is therefore distinct from the conifold locus. Keeping the $u$-coordinate cover distinguishes these two points and retains the regular point $u=0$.

For the tensor calculation, write
\begin{equation}
 D=u^2-1,\qquad H=1+2\epsilon u+\epsilon^2D,\qquad {k_0}=\frac{i}{\pi^3}.
\end{equation}
Transforming the full tensor \eqref{eq:CYcouplings} to $(u,\epsilon)$, including the symmetric-tensor multiplicities, gives
\begin{align}\label{g:allcomponents}
 C_{uuu}&={k_0}\frac{\epsilon[1+4\epsilon u+\epsilon^2(1+3u^2)]}{2DH^2},\notag\\
 C_{\epsilon uu}&={k_0}\frac{\epsilon(1+\epsilon u)}{H^2},\notag\\
 C_{\epsilon\epsilon u}&={k_0}\frac{\epsilon^2D-1}{2\epsilon H^2},\notag\\
 C_{\epsilon\epsilon\epsilon}&=-{k_0}\frac{u+\epsilon D}{\epsilon^2H^2}.
\end{align}
Here
\begin{equation}
 \mathcal C=C_{uuu}(du)^3+3C_{\epsilon uu}d\epsilon(du)^2
 +3C_{\epsilon\epsilon u}(d\epsilon)^2du
 +C_{\epsilon\epsilon\epsilon}(d\epsilon)^3.
\end{equation}
The identities hold exactly on the smooth part of the punctured chart, before the limit $\epsilon\to0$ is taken.

\begin{proposition}[Leading tangential coefficient]\label{prop:rigidlimit}
On a branch with $u=\varphi^{-1/2}$, the leading tangential coefficient on the fixed-$\epsilon$ slices of \eqref{eq:xyexplicit} is
\begin{equation}\label{eq:Crig}
 \Yrig_{\varphi\varphi\varphi}:=
 \lim_{\epsilon\to0}\epsilon^{-1}C_{\varphi\varphi\varphi}
 =-\frac{i}{16\pi^3(1-\varphi)\varphi^{7/2}},
 \qquad u=\varphi^{-1/2}.
\end{equation}
The convergence is uniform on compact subsets of the chosen finite $\varphi$-coordinate domain away from $\varphi=0,1$.
\end{proposition}
\begin{proof}
The first component of \eqref{g:allcomponents} gives
\begin{equation*}
 C_{uuu}=\frac{i\epsilon}{2\pi^3(u^2-1)}+O(\epsilon^3).
\end{equation*}
The linear terms in the numerator and in $H^2$ cancel, so there is no $\epsilon^2$ term. The change of variable $du/d\varphi=-\tfrac12\varphi^{-3/2}$ now gives \eqref{eq:Crig}. On compact subsets of $B$, $D$ stays away from zero and $H\to1$ uniformly. This gives the uniform expansion and hence the asserted convergence.
\end{proof}

Write $\Yrig=\Yrig_{\varphi\varphi\varphi}(d\varphi)^3$. Its expression in $\varphi$ requires a branch, but the cubic differential itself is holomorphic on $B$ and extends meromorphically to $\PP^1_u$. It is the leading Taylor coefficient in the chosen normal coordinate, rather than a Poincar\'e residue. To understand its dependence on the transverse slices, we must also keep the mixed components in \eqref{g:allcomponents}.

\section{Reconstructing the Schwarzian from the leading logarithmic Yukawa tensor}\label{sec:comparison}
The tangential Yukawa coefficient is proportional to the holomorphic Seiberg--Witten cubic tensor. To recover the Schwarzian quadratic differential $\Quad$, we use the mixed normal components to select a horizontal splitting of the logarithmic tangent sequence along the open exceptional divisor. Together with the prescribed constant-coefficient connection rule, this splitting gives two projective connections whose difference is $\Quad$. The construction is independent of the Hodge and normal frames.

Put $\theta=d\epsilon/\epsilon$. In the logarithmic cotangent frame $(du,\theta)$, the tensor has the expansion $\mathcal C=\epsilon\Tens+O(\epsilon^2)$, with
\begin{equation}\label{g:fullsymbol}
 \Tens={k_0}\left[
 \frac{(du)^3}{2(u^2-1)}-\frac32du\,\theta^2-u\theta^3
 \right].
\end{equation}
The tensor $\Tens$ is the first nonzero normal coefficient of $\mathcal C$ in the logarithmic cotangent bundle; we call it the \emph{leading logarithmic Yukawa tensor}. There is no $du^2\theta$ term, while the $du\theta^2$ and $\theta^3$ terms remain at this order. The logarithmic frame puts these components at the same leading order, even though their orders differ in the ordinary frame $(du,d\epsilon)$, as shown in \eqref{g:allcomponents}.

To describe this coefficient independently of its scalar components, let $X$ be a neighborhood of the open exceptional divisor $B=\PP^1_u\setminus\{-1,1,\infty\}$ in the blowup chart of the double cover. Its boundary is $B=\{\epsilon=0\}$, with ideal sheaf $I$, and the class of $\epsilon$ generates the conormal line $I/I^2$ locally. We choose an extension $\mathscr L$ of the Hodge line in which the prescribed algebraic Hodge frame remains a local nonvanishing generator. We use this extension without identifying it with a canonical limiting Hodge extension. The leading logarithmic Yukawa tensor is then a section of
\begin{equation}\label{eq:value-line}
 (I/I^2)\otimes\mathscr L^{-2}|_B
 \otimes\operatorname{Sym}^3(\Omega_X^1(\log B)|_B).
\end{equation}
Here $\Omega_X^1(\log B)$ is the bundle of logarithmic one-forms, locally generated by $du$ and $d\epsilon/\epsilon$, and $\operatorname{Sym}^3$ denotes its third symmetric power. The factor $\mathscr L^{-2}$ is the square of the dual Hodge line. We write $M=(I/I^2)\otimes\mathscr L^{-2}|_B$ for the \emph{coefficient line bundle}. In its frame $[\epsilon]\otimes\Omega^{-2}$, the leading tensor has the components \eqref{g:fullsymbol}.

Throughout this section, frame changes are by holomorphic units along $B$; a rescaling with a zero or pole there would change the prescribed extension. Under $\Omega\mapsto h\Omega$, the scalar components acquire a factor $h_0^2$, where $h_0=h|_B$. A change of normal coordinate $\widetilde\epsilon=c(u)\epsilon$ contributes the inverse factor $c(u)^{-1}$ and changes the logarithmic cotangent frame as well. These are the transformation laws of the same $M$-valued tensor.

\subsection{Normal coordinate changes and the rigid special K\"ahler cubic tensor}
Let $\widetilde\epsilon=c(u)\epsilon$, where $c$ is a nowhere-vanishing holomorphic function, and put $L=c'/c$, with $c'=dc/du$. Holding $\widetilde\epsilon$ fixed gives the tangent vector
\begin{equation}
 \partial_u\big|_{\widetilde\epsilon}
 =\partial_u\big|_\epsilon-\epsilon L\partial_\epsilon.
\end{equation}
Evaluating the full tensor \eqref{g:allcomponents} on three copies of this vector and extracting the coefficient of $\widetilde\epsilon$ gives
\begin{equation}\label{g:slicing}
 \widetilde C^{\rm rig}
 =\frac{{k_0}}{c(u)}\left[
 \frac1{2(u^2-1)}-\frac32L^2+uL^3
 \right](du)^3.
\end{equation}
The last two terms record the change of slices. For example, on $u\ne0$, the choice $c(u)=u$ gives ${k_0}(du)^3/[2u^3(u^2-1)]$. The mixed components therefore contribute to the tangential coefficient, in addition to the factor $c^{-1}$.

We compare the tangential cubic differential on the original slices with the rigid special K\"ahler cubic, to which it is proportional. Below, the leading logarithmic Yukawa tensor will recover the corresponding horizontal splitting of the logarithmic tangent sequence along $B$. This splitting is defined only along the boundary and does not specify a transverse foliation away from it.

The holomorphic cubic tensor of rigid special K\"ahler geometry is
\begin{equation}\label{eq:YSW}
 \YSW=(da)^2d\tau
  =\big(a_z(a_D)_{zz}-(a_D)_z a_{zz}\big)(dz)^3,
\end{equation}
where $a_z=\partial a/\partial z$, and the other derivatives are denoted similarly. In the special coordinate $a$, the tensor is $\mathcal F'''(a)(da)^3$. The determinant expression cancels the extra terms from a change of source coordinate and also shows invariance under constant determinant-one changes of the electric-magnetic frame.

\begin{theorem}[Comparison of the holomorphic cubic tensors]\label{thm:cubic}
With the Calabi--Yau Hodge frame and elliptic period normalization fixed as above,
\begin{equation}\label{eq:cubicmain}
 \Yrig=\frac{2}{\pi^2}\YSW.
\end{equation}
In particular, on the $u$-coordinate cover of the rigid-limit divisor,
\begin{equation}\label{eq:cubicu}
 \Yrig_{uuu}=\frac{i}{2\pi^3(u^2-1)},\qquad
 (\YSW)_{uuu}=\frac{i}{4\pi(u^2-1)}.
\end{equation}
\end{theorem}
\begin{proof}
The change of variable $d\varphi/du=-2u^{-3}$ in \eqref{eq:Crig} gives the first identity in \eqref{eq:cubicu}; its extension through $u=0$ follows directly from \eqref{g:allcomponents}. The second identity follows on multiplying the second identity of \eqref{eq:Wu} by $(\partial_u a)^2$. Comparing the two gives \eqref{eq:cubicmain}. The same comparison can be made in the $\varphi$ coordinate, where \eqref{eq:Wvarphi} gives
\begin{equation}\label{eq:Cwronskian}
 \Yrig_{\varphi\varphi\varphi}
 =\frac{(\varpi_0')^2\tau'}{\pi^2}
 =\frac{2}{\pi^2}(\partial_\varphi a)^2\partial_\varphi\tau.
\end{equation}
\end{proof}

\subsection{The horizontal splitting}
The splitting takes place in the logarithmic tangent bundle $T_X(-\log B)$, the dual of the logarithmic cotangent bundle. Locally its generators are $\partial_u$ and $\epsilon\partial_\epsilon$, and we write its restriction to $B$ as
\begin{equation}
 \mathscr A=T_X(-\log B)|_B.
\end{equation}
It fits into the exact sequence
\begin{equation}\label{g:logsequence}
 0\longrightarrow\mathcal O_B E\longrightarrow\mathscr A
 \longrightarrow T_B\longrightarrow0.
\end{equation}
Here $T_B$ is the tangent line bundle of $B$, $\mathcal O_B$ is its sheaf of holomorphic functions, and $E$ is the vertical Euler section. Locally $E=\epsilon\partial_\epsilon$, with $\theta(E)=1$ and $du(E)=0$. Its class in $\mathscr A$ is nonzero, even though the ordinary vector field vanishes on $B$. A \emph{horizontal splitting of the logarithmic tangent sequence}, or simply a \emph{horizontal splitting}, is a line subbundle of $\mathscr A$ that maps isomorphically to $T_B$.

Contracting the leading Yukawa tensor with $E$ gives an $M$-valued symmetric bilinear form
\begin{equation}
 B_E(v,w)=\Tens(E,v,w).
\end{equation}
In the frame $(\partial_u,E)$ its matrix is
\begin{equation}\label{g:polar}
 B_E={k_0}\begin{pmatrix}0&-1/2\\-1/2&-u\end{pmatrix}.
\end{equation}
Its determinant is $-{k_0}^2/4\ne0$, so the form is nondegenerate throughout $B$. A line is called $B_E$-null when $B_E(v,v)=0$ for every vector $v$ in the line.

\begin{proposition}[Global selection of the horizontal line]\label{g:horizontal}
There is a unique holomorphic $B_E$-null line subbundle of $\mathscr A$ complementary to $\mathcal O_BE$ at every point of $B$. In the chosen chart, it is
\begin{equation}
 \mathscr H=\Span(\partial_u).
\end{equation}
Thus the leading tensor and the logarithmic boundary structure determine this splitting on a boundary cover that includes $u=0$.
\end{proposition}
\begin{proof}
For $v=A\partial_u+BE$, the null equation is
\begin{equation}
 \Tens(E,v,v)=-{k_0} B(A+uB)=0.
\end{equation}
The two null line subbundles are therefore
\begin{equation}
 \mathscr H_0=\Span(\partial_u),\qquad
 \mathscr H_1=\Span(u\partial_u-E).
\end{equation}
These lines are distinct at every point. The first maps isomorphically to $T_B$, while the second becomes vertical at $u=0$. A holomorphic null line on the connected base $B$ must stay on one of these two disjoint branches. Only $\mathscr H_0$ is therefore complementary to the vertical line everywhere. Since this characterization uses only the tensor and transversality, it is independent of coordinates and of the frame of $M$.
\end{proof}

Including $u=0$ is essential for uniqueness: on a small chart avoiding that point, both null lines are horizontal complements.

Let $h:T_B\to\mathscr A$ be the resulting horizontal lift, with image $\mathscr H$. For this lift, define the contractions
\begin{equation}\label{g:threecontractions}
 \ell=\Tens(E,E,E),\qquad
 \beta(v)=\Tens(E,E,h(v)),\qquad
 C(v_1,v_2,v_3)=\Tens(h(v_1),h(v_2),h(v_3)).
\end{equation}
They give an $M$-valued function $\ell$, an $M$-valued one-form $\beta$, and an $M$-valued cubic differential $C$. In the chosen frame, their components are
\begin{equation}\label{g:contractionvalues}
 \ell=-{k_0}u,\qquad \beta=-\frac {k_0}2du,\qquad
 C=\frac{{k_0}}{2(u^2-1)}(du)^3.
\end{equation}

\subsection{The prescribed connection and projective reconstruction}
The section $\ell$ has one simple zero, at $p_0=\{u=0\}$, and $\beta$ is nowhere zero. The differential of a line-bundle section at one of its zeros is independent of the chosen frame; here it belongs to $M_{p_0}\otimes T_{p_0}^*B$. We can therefore define the scalar
\begin{equation}\label{g:cnumber}
  c=\left.\frac{d\ell}{\beta}\right|_{p_0}=2,
\end{equation}
which is independent of both coordinate and frame.

\begin{lemma}[The connection prescribed on $M$]\label{g:lineconnection}
The rule
\begin{equation}\label{g:connection}
 D_M\ell=c\beta
\end{equation}
defines a unique holomorphic connection on $M|_{B\setminus\{p_0\}}$. This connection extends holomorphically across $p_0$ and is flat on $B$.
\end{lemma}
\begin{proof}
In a local frame $e$ of $M$, write $\ell=l e$ and $\beta=b\,dz\,e$. If $D_Me=A\otimes e$, with $A$ a one-form, the Leibniz rule gives the connection form required by \eqref{g:connection}:
\begin{equation}
 A=\frac{cb\,dz-dl}{l}.
\end{equation}
The definition of $c$ in \eqref{g:cnumber} makes the numerator vanish at the simple zero of $l$, so the quotient extends holomorphically. In the frame used in \eqref{g:contractionvalues}, $A=0$. Under a change of scalar components $(l,b)\mapsto g(l,b)$, the form changes as $A\mapsto A-d\log g$. This is precisely the connection transformation law for the inverse change of frame, so the local connections glue. Finally, a holomorphic connection on a line bundle over a complex curve has zero curvature and is therefore flat.
\end{proof}

We require $c$ to be constant as part of the prescription. Within this prescription, regularity at $p_0$ determines the constant, but regularity alone does not single out a connection among all holomorphic connections on $M$. If $c$ were allowed to vary with the base point, other rules could have the same value at $p_0$.

Now choose a $D_M$-parallel frame on a simply connected coordinate neighborhood. Taking the scalar coefficients of $\beta$ and $C$ in this frame, define local primitives by
\begin{equation}\label{g:UV}
 dU=\beta,\qquad dV=\frac{C}{\beta^2}.
\end{equation}
The quotient $C/\beta^2$ is intrinsically an $M^{-1}$-valued one-form. If the scalar coefficients in the parallel frame are $\beta=b\,dz$ and $C=c_3(dz)^3$, its scalar coefficient is $(c_3/b^2)\,dz$. Both one-forms are nowhere zero on $B$, so the primitives are locally univalent. Since two parallel frames differ by a nonzero constant, a frame change rescales the two forms by reciprocal constants and changes their primitives affinely.

\begin{theorem}[Reconstruction from the leading logarithmic Yukawa tensor]\label{g:reconstruction}
On the open exceptional divisor $B$, the leading tensor \eqref{g:fullsymbol}, together with the constant-coefficient rule \eqref{g:connection}, determines the quadratic differential
\begin{equation}\label{g:QT}
 \Quad_{\Tens}=\bigl(\Sch{U}{z}-\Sch{V}{z}\bigr)(dz)^2.
\end{equation}
This differential is independent of the source coordinate, the frame of $M$, the chosen parallel frame, and the integration constants. It is also unchanged by nowhere-vanishing changes of the Hodge and normal frames, with the leading tensor and splitting transformed accordingly. Explicitly,
\begin{equation}
 \Quad_{\Tens}=-\frac{2(du)^2}{(u^2-1)^2}
 =\Sch{a_D/a}{\tau}(d\tau)^2.
\end{equation}
\end{theorem}
\begin{proof}
The splitting of Proposition~\ref{g:horizontal} and the prescribed connection of Lemma~\ref{g:lineconnection} are independent of coordinates. Changing the parallel frame or integration constants acts affinely on $U,V$, leaving their Schwarzians unchanged. A change of source coordinate contributes the same extra Schwarzian term to both, which cancels in their difference. Thus \eqref{g:QT} is a quadratic differential.

In the displayed frame $D_M=d$, and \eqref{g:contractionvalues} gives
\begin{equation}
 U=-\frac {k_0}2u+\text{constant},\qquad
 V=\frac1{k_0}\log\frac{u-1}{u+1}+\text{constant}.
\end{equation}
Differentiation gives $\Sch{\log((u-1)/(u+1))}{u}=2/(u^2-1)^2$, proving the stated formula for $\Quad_{\Tens}$.

To identify this expression with the Seiberg--Witten quadratic differential, apply Theorem~\ref{thm:ranktwo} to $s=(a,a_D)$. By Lemma~\ref{lem:wronskian} and Theorem~\ref{thm:cubic}, $W=(i/\pi)du$ and $\Ygen=\YSW=i(du)^3/[4\pi(u^2-1)]$. The primitives of $W$ and $\Ygen/W^2$ differ from $U,V$ only by affine transformations. The rank-two identity therefore identifies $\Quad_{\Tens}$ with \eqref{eq:Qintro}.
\end{proof}

\section{Genus-zero Gromov--Witten invariants of the compact Calabi--Yau threefold}\label{sec:enumerative}
Base degrees one and two provide the compact data needed for the Schwarzian comparison. We compute their genus-zero generating functions from the Picard--Fuchs system, the mirror map on $y=0$, and the full Yukawa tensor, using the genus-zero mirror correspondence as an input. We then locate their dominant singularity and prove the coefficient asymptotics as the fiber degree tends to infinity. The curve-class and multiple-cover conventions are fixed below.

\subsection{Generating functions and conventions}
Let $X_{\mathrm{CY}}$ be a smooth crepant resolution of a generic degree-$12$ hypersurface in $\PP^4_{1,1,2,2,6}$. Fix the large-radius chamber of \cite[Section~2]{Candelas}. Let $H_{\mathrm{CY}}$ be the divisor class associated with hypersurfaces of weighted degree two, and let $L_{\mathrm{CY}}$ be the class of a fiber of the K3 fibration. They satisfy
\begin{equation}\label{e:intersectionmarking}
 H_{\mathrm{CY}}^3=4,\qquad H_{\mathrm{CY}}^2L_{\mathrm{CY}}=2,
 \qquad H_{\mathrm{CY}}L_{\mathrm{CY}}^2=L_{\mathrm{CY}}^3=0.
\end{equation}
Let $\gamma_{\mathrm f},\gamma_{\mathrm b}$ be the dual integral curve classes, so that
$H_{\mathrm{CY}}\cdot\gamma_{\mathrm f}=L_{\mathrm{CY}}\cdot\gamma_{\mathrm b}=1$ and
$H_{\mathrm{CY}}\cdot\gamma_{\mathrm b}=L_{\mathrm{CY}}\cdot\gamma_{\mathrm f}=0$.
Thus $\beta_{j,k}=j\gamma_{\mathrm f}+k\gamma_{\mathrm b}$ has degrees
$j=H_{\mathrm{CY}}\cdot\beta_{j,k}$ and $k=L_{\mathrm{CY}}\cdot\beta_{j,k}$.
Here $k$ is the degree over the base of the K3 fibration; we call the $H_{\mathrm{CY}}$-degree $j$ the fiber degree in this basis. The rational genus-zero Gromov--Witten invariants are
\begin{equation}\label{e:GWdefinition}
 N^0_{j,k}=\int_{[\overline{\mathcal M}_{0,0}(X_{\mathrm{CY}},\beta_{j,k})]^{\mathrm{vir}}}1,
 \qquad j\ge0,\quad k\ge1.
\end{equation}
The stack $\overline{\mathcal M}_{0,0}(X_{\mathrm{CY}},\beta_{j,k})$ parametrizes genus-zero stable maps with no marked points and curve class $\beta_{j,k}$, and $[\ ]^{\mathrm{vir}}$ denotes its virtual fundamental class. Its virtual dimension is zero, so the integral is a rational number. These invariants of the compact threefold $X_{\mathrm{CY}}$ are computed using the periods and Yukawa tensor of the mirror family in Section~\ref{sec:mirror}.

The genus-zero instanton numbers $n^0_{j,k}$ are defined by the multiple-cover relation
\begin{equation}\label{e:GWmulticover}
 N^0_{j,k}=\sum_{\substack{d\ge1\\d\mid j,\ d\mid k}}
       \frac{n^0_{j/d,k/d}}{d^3}.
\end{equation}
As usual, every positive integer divides $0$. Equation~\eqref{e:GWmulticover} defines the $n^0_{j,k}$ recursively \cite[(1)]{BKL}; we do not require their integrality. We use the genus-zero mirror correspondence in the chosen chamber as an input. In the conventions of \cite{Candelas,HKTY,HM}, it identifies the Gromov--Witten potential with the normalized mirror prepotential. Let $T$ and $\Sigma$ be the flat coordinates dual to fiber and base degree, respectively, and set
\begin{equation}\label{e:nomes}
\begin{gathered}
 Q=e^{2\pi iT},\qquad Q_2=e^{2\pi i\Sigma},\\
 D_k(Q)=\sum_{j\ge0}n^0_{j,k}Q^j,\qquad
 D_k^{\mathrm{GW}}(Q)=\sum_{j\ge0}N^0_{j,k}Q^j.
\end{gathered}
\end{equation}
In these flat coordinates, write the compact prepotential as
\begin{equation}\label{e:CYF}
 \mathscr F_{\mathrm{CY}}=\Sigma T^2+f_0(T)
  +\sum_{k\ge1}F_k(T)Q_2^k,
\end{equation}
where the contribution from classes of positive base degree is
\begin{equation}\label{e:GWpotential}
 \sum_{k\ge1}F_k(T)Q_2^k
 =\frac1{(2\pi i)^3}\sum_{\substack{j\ge0\\k\ge1}}
 N^0_{j,k}Q^jQ_2^k
 =\frac1{(2\pi i)^3}\sum_{\substack{j\ge0\\k\ge1}}
 n^0_{j,k}\operatorname{Li}_3(Q^jQ_2^k).
\end{equation}
The function $f_0(T)$ contains the terms of base degree zero, including the classical term $2T^3/3$. Together with $\Sigma T^2$, this gives the classical cubic $\frac16(T H_{\mathrm{CY}}+\Sigma L_{\mathrm{CY}})^3$ in \eqref{e:CYF}, as prescribed by \eqref{e:intersectionmarking}. Expanding $\operatorname{Li}_3(w)=\sum_{d\ge1}w^d/d^3$ for $|w|<1$ gives the first two coefficients of positive base degree:
\begin{equation}\label{e:cover}
\begin{aligned}
 (2\pi i)^3F_1&=D_1^{\mathrm{GW}}(Q)=D_1(Q),\\
 (2\pi i)^3F_2&=D_2^{\mathrm{GW}}(Q)=D_2(Q)+\frac18D_1(Q^2).
\end{aligned}
\end{equation}
This agrees with the normalization in \cite[(1.1)--(1.2)]{HM}. The functions $F_k(T)$ are coefficients of the compact prepotential, whereas $f_n$ in \eqref{eq:F} are gauge-instanton coefficients. The Schwarzian coefficients $A_n$ and the individual Gromov--Witten invariants $N^0_{j,k}$ are distinct quantities.

We shall show that the generating functions have their dominant singularity at $Q=\rho$, where $\rho=e^{-2\pi}$, and that
\begin{equation}\label{e:countlimit}
 \lim_{j\to\infty}
 \frac{e^{2\pi j}n^0_{j,2}}{j^3(n^0_{j,1})^2}=\frac1{384}.
\end{equation}
For the Schwarzian in the fixed electric-magnetic marking,
\begin{equation}\label{e:Schdefs}
 S(\tau)=\Sch{a_D/a}{\tau}
   =32\pi^2\sum_{n\ge1}A_nq^n,\qquad q=e^{\pi i\tau},
\end{equation}
the corresponding relation is
\begin{equation}\label{e:directrelation}
 \frac{A_2}{A_1^2}
 =46080\lim_{j\to\infty}
 \frac{e^{2\pi j}N^0_{j,2}}{j^3(N^0_{j,1})^2}-112=8.
\end{equation}
We derive the generating functions and their asymptotics below. The normalized rigid limit in Section~\ref{sec:enumerative-rigid} then relates them to the Schwarzian and proves \eqref{e:directrelation}. At base degree two, the multiple-cover term grows exponentially more slowly than the leading term, so replacing $n^0_{j,k}$ by $N^0_{j,k}$ leaves the limit in \eqref{e:countlimit} unchanged.

\subsection{The mirror map on the boundary \texorpdfstring{$y=0$}{y=0}}
In the heterotic interpretation, $y=0$ is the weak-coupling boundary of the mirror family \cite{Kachru}. It is distinct from the Seiberg--Witten weak-coupling limit $u\to\infty$. The algebraic moduli $x,y$ in \eqref{eq:xy} are related to the Frobenius coordinates by $x=1728\zeta_1$, $y=4\zeta_2$; the notation $\zeta_1,\zeta_2$ avoids confusion with the homogeneous coordinates of the hypersurface. With $\theta_x=x\partial_x$ and $\theta_y=y\partial_y$, the normalized periods satisfy $\mathcal L_1\Pi=\mathcal L_2\Pi=0$, where
\begin{align}\label{e:PF}
 \mathcal L_1&=\theta_x^2(\theta_x-2\theta_y)
 -x(\theta_x+\tfrac16)(\theta_x+\tfrac12)(\theta_x+\tfrac56),\notag\\
 \mathcal L_2&=\theta_y^2
 -\frac y4(2\theta_y-\theta_x)(2\theta_y-\theta_x+1).
\end{align}
These are the operators of \cite[Appendix A.1, (A.6)--(A.7)]{HKTY}, after the coordinate rescaling above. The residue-period normalization there includes the factor $a_0=-12\psi$, while the three-form of \cite[Section~6.1]{Candelas} includes $\psi$. The two Hodge frames therefore differ by a constant. We retain the Yukawa normalization \eqref{eq:CYcouplings} and choose the fundamental period so that $\omega_{\mathrm{CY}}(0,0)=1$. We write $\omega_{\mathrm{CY}}$ to distinguish it from the Legendre periods $\omega_0,\omega_1$.

Define
\begin{align}\label{e:Pg}
 a_m&=\frac{(1/6)_m(1/2)_m(5/6)_m}{(m!)^3},\notag\\
 P(x)&=\sum_{m\ge0}a_mx^m
 ={}_3F_2(\tfrac16,\tfrac12,\tfrac56;1,1;x),\notag\\
 g(x)&=\frac{2}{P(x)}\sum_{m\ge1}a_mH_mx^m,
 \qquad H_m=\sum_{r=1}^m\frac1r,\quad H_0=0.
\end{align}
The coefficient of $\zeta_1^m\zeta_2^n$ in the fundamental Frobenius solution is
\begin{equation}
 \frac{(6m)!}{(3m)!(m!)^2(n!)^2(m-2n)!},\qquad m\ge2n.
\end{equation}
The coefficient vanishes for $m<2n$. Differentiating the Frobenius series with respect to its second exponent parameter gives the logarithmic base period. At $n=0$, the holomorphic part has the fundamental-period coefficient multiplied by $2H_m$. Using the $n=1$ recurrence from $\mathcal L_2$, we obtain
\begin{equation}\label{e:mirrorexp}
 \omega_{\mathrm{CY}}=P+\frac y4x^2P''+O(y^2),\qquad
 \log Q_2=\log(y/4)+g(x)+O(y).
\end{equation}
In this subsection and the next, primes on functions of $x$ denote derivatives with respect to $x$. These expansions hold in $y$ at fixed $x$; no uniformity as $x\to1$ is asserted. The simultaneous limit $x\to1$, $y\to0$ requires the separate argument in Section~\ref{sec:enumerative-rigid}.

We use the normalized Eisenstein series
\begin{equation*}
 E_2(T)=1-24\sum_{n\ge1}\sigma_1(n)Q^n,\qquad
 E_6(T)=1-504\sum_{n\ge1}\sigma_5(n)Q^n.
\end{equation*}
The normalization of $E_4$ is as in \eqref{eq:Edef}. The function $E_2$ is quasimodular, while $E_4$ and $E_6$ are modular forms of weights four and six. Write $j(T)=1728E_4(T)^3/(E_4(T)^3-E_6(T)^2)$ for the classical modular invariant; the function $j(T)$ is unrelated to the degree index $j$.
We will use the Ramanujan identities \cite[Proposition~15]{Zagier}, with $\vartheta=(2\pi i)^{-1}\partial_T$:
\begin{equation}\label{e:Ramanujan}
 \vartheta E_2=\frac{E_2^2-E_4}{12},\quad
 \vartheta E_4=\frac{E_2E_4-E_6}{3},\quad
 \vartheta E_6=\frac{E_2E_6-E_4^2}{2}.
\end{equation}
On $y=0$, the logarithmic Frobenius marking fixes
$T(x)=(2\pi i)^{-1}\log(x/1728)+O(x)$. In this normalization,
\begin{equation}\label{e:modmirror}
 x=\frac{1728}{j(T)},\qquad P(x)^2=E_4(T),\qquad
 1-x=\frac{E_6(T)^2}{E_4(T)^3}.
\end{equation}
The first identity is the weak-coupling mirror map of \cite[(1.3)]{Kachru}. To check that it has the stated Frobenius normalization, put
$F(x)={}_2F_1(1/12,5/12;1;x)$. Clausen's formula gives $P=F^2$ \cite[(16.12.2)]{DLMF}. The restriction of $\mathcal L_1$ to $y=0$ is the symmetric square of
\[
 x(1-x)F''+(1-3x/2)F'-\frac5{144}F=0;
\]
so its solution space is spanned by products of pairs of solutions of this second-order equation. The logarithmic period $PT$ is therefore $F$ times a second solution; the freedom to add a multiple of $P$ is fixed by the cusp normalization. Abel's identity then gives
\begin{equation}\label{e:Tprime}
 T'(x)=\frac1{2\pi i\,x\sqrt{1-x}\,P(x)},
\end{equation}
where $\sqrt{1-x}=1+O(x)$ at the cusp. The classical identity in \cite[Section~5.4, (74)]{Zagier} gives
$F(1728/j(T))=E_4(T)^{1/4}$ near $Q=0$. Equation~\eqref{e:Ramanujan} gives
$\vartheta x=xE_6/E_4$ for $x=1728/j(T)$. Its inverse has derivative \eqref{e:Tprime}, and $j(T)=Q^{-1}+744+O(Q)$ fixes the additive constant to be the prescribed one. The inverse therefore agrees with the Frobenius coordinate, proving \eqref{e:modmirror} in the stated normalization.
We fix the branches by requiring $0<x<1$ to correspond to $T$ on the positive imaginary axis with $\operatorname{Im}T>1$. The endpoint $x=1$ corresponds to $T=i$ and $Q=\rho$; this is the gauge-symmetry enhancement point.

Set $\mathcal E=E_4(i)$ and $g_*=g(1)$. The value $g_*$ is finite, with the integral representation
\begin{equation}\label{e:gstar}
 g_*=\frac{2}{P(1)}\int_0^1\frac{P(1)-P(t)}{1-t}\,dt,
 \qquad P(1)=\sqrt{\mathcal E}=\frac{\sqrt3\,\Gamma(1/4)^4}{8\pi^3}.
\end{equation}
The value of $E_4(i)$ is recorded in \cite[Section~6.3]{Zagier}. Stirling's formula gives $a_m=O(m^{-3/2})$, so $\sum_m a_mH_m$ converges. Substituting $H_m=\int_0^1(1-t^m)/(1-t)\,dt$ into the series for $g_*$ and applying monotone convergence proves \eqref{e:gstar}. The endpoint is integrable because $P(1)-P(t)=O(\sqrt{1-t})$. We use this exact expression for $g_*$ below.

\subsection{Generating functions in base degrees one and two}
Throughout this calculation, the Eisenstein series are evaluated at $T$, and $g$ denotes $g(x(T))$. Write $\vartheta=Q\partial_Q=(2\pi i)^{-1}\partial_T$.

\begin{theorem}[Generating functions at base degrees one and two]\label{e:Dtheorem}
The generating functions in base degrees one and two are
\begin{equation}\label{e:D1}
 D_1(Q)=\frac{2e^{-g(x)}}{(1-x)P(x)^2}
              =2e^{-g(x)}\frac{E_4^2}{E_6^2},
\end{equation}
and
\begin{align}\label{e:D2}
 D_2(Q)+\frac18D_1(Q^2)=D_1(Q)^2\left[
 -\frac14\vartheta^2g+\frac{E_2^2}{72}
 +\frac{E_2E_6}{18E_4}\right.\notag\\[-2pt]
 \left.\hspace{20mm}+\frac{5E_4^4}{64E_6^2}
 -\frac{37E_4}{192}+\frac{31E_6^2}{288E_4^2}\right].
\end{align}
The identities hold as convergent germs at $Q=0$. They continue analytically along the chosen period branch wherever the expressions on the right are defined.
\end{theorem}

\begin{proof}
Use the logarithmic one-forms $X=d\log x$ and $Y=d\log y$, and put $\Delta=(1-x)^2-x^2y$. We write $\widehat{\mathcal C}=(2\pi i)^3\mathcal C$ and use subscripts $X,Y$ for its components in this frame. Transforming \eqref{eq:CYcouplings} gives
\begin{align}\label{e:rationaltensor}
 \widehat C_{XXX}&=\frac4\Delta,&
 \widehat C_{XXY}&=\frac{2(1-x)}\Delta,\notag\\
 \widehat C_{XYY}&=\frac{y(2x-1)}{(1-y)\Delta},&
 \widehat C_{YYY}&=\frac{y[(1-x)-y(3x-1)]}{2(1-y)^2\Delta}.
\end{align}
The associated symmetric cubic tensor is $\widehat C_{XXX}X^3+3\widehat C_{XXY}X^2Y+3\widehat C_{XYY}XY^2+\widehat C_{YYY}Y^3$. Define
\begin{equation}\label{e:Y222definition}
 \mathcal Y_{222}:=\left.\frac{\partial^3\mathscr F_{\mathrm{CY}}}{\partial\Sigma^3}\right|_T
 =\sum_{k\ge1}k^3D_k^{\mathrm{GW}}(Q)Q_2^k.
\end{equation}
The index $2$ refers to the base coordinate $\Sigma$. To obtain this component, evaluate $\widehat{\mathcal C}$ on three copies of $\partial_{\log Q_2}$ at fixed $T$ and divide by $\omega_{\mathrm{CY}}^2$. Since $\partial_\Sigma=2\pi i\,\partial_{\log Q_2}$, the factor $(2\pi i)^3$ is already included in $\widehat{\mathcal C}$. The contraction implements the coordinate change, and division by the fundamental period squared gives the Yukawa tensor in the period-normalized Hodge frame.

At first order, $\partial_{\log Q_2}\log x=O(y)$ and $\partial_{\log Q_2}\log y=1+O(y)$. Only $\widehat C_{YYY}=y/[2(1-x)]+O(y^2)$ can therefore contribute to the coefficient of $Q_2$. Substituting $y=4e^{-g}Q_2+O(Q_2^2)$ into \eqref{e:Y222definition} gives \eqref{e:D1}.

For the coefficient of $Q_2^2$, put $s=1-x$, $L=P'/P$, and $V_P=P''/P$. Applying the $\mathcal L_2$ recurrence to the period $PT$ gives
\begin{equation}
 T(x,y)=T(x)+\frac y4x^2(T''+2LT')+O(y^2).
\end{equation}
Keeping the flat coordinate $T$ fixed thus requires $x_{\mathrm{alg}}=x+yk(x)+O(y^2)$, where the unadorned $x$ denotes the value at $y=0$. The correction is
\begin{equation}\label{e:kdef}
 k=-\frac{x^2}4\left(\frac{T''}{T'}+2L\right)
   =\frac x4-\frac{x^2}{8s}-\frac{x^2L}4.
\end{equation}
The correction to the base mirror map comes from the logarithmic period. Write
\begin{align*}
 \omega_{\mathrm{CY}}&=\sum_{n\ge0}p_n(x)y^n,\\
 2\pi i\Pi_\Sigma&=\omega_{\mathrm{CY}}\log(y/4)+\sum_{n\ge0}g_n(x)y^n.
\end{align*}
Here $p_0=P$, $g_0=Pg$, and $\Sigma=\Pi_\Sigma/\omega_{\mathrm{CY}}$.
The operator $\mathcal L_2$ gives, for $n\ge1$,
\begin{equation}\label{e:periodrecurrence}
 n^2p_n=\frac14(2n-2-\theta_x)(2n-1-\theta_x)p_{n-1},\qquad
 p_n=\frac{x^{2n}P^{(2n)}}{4^n(n!)^2}.
\end{equation}
Applying the same operator to the logarithmic period gives additional terms when it differentiates the explicit logarithm. At order $y$, the resulting recurrence gives
\begin{equation}\label{e:logperiodrecurrence}
 g_1=\frac{x^2}{4}(Pg)''-\frac{x^2}{2}P''-xP'+\frac12P.
\end{equation}
Dividing the logarithmic series by $\omega_{\mathrm{CY}}$ and then substituting $x+yk$ for $x$ gives the correction at fixed $T$:
\begin{equation}\label{e:hdef}
 h=\frac{x^2}4g''+
 \left(\frac x4-\frac{x^2}{8s}+\frac{x^2L}4\right)g'
 -xL+\frac12-\frac{x^2V_P}2,
\end{equation}
so that, along this fixed-$T$ curve,
\begin{align}\label{e:fixedTexp}
 Q_2&=\frac{e^g}4y(1+yh+O(y^2)),\notag\\
 \partial_{\log Q_2}\log x_{\rm alg}&=y\frac{k}{x}+O(y^2),\qquad
 \partial_{\log Q_2}\log y=1-yh+O(y^2),\notag\\
 \omega_{\mathrm{CY}}^{-2}&=P^{-2}\left[1-2y(kL+x^2V_P/4)+O(y^2)\right].
\end{align}
At fixed algebraic $x$, the coefficient of $y$ in $\log Q_2$ is $x^2(g''+2Lg')/4-xL+1/2-x^2V_P/2$. The change $x\mapsto x+yk$ needed to hold $T$ fixed adds $kg'$, giving \eqref{e:hdef}.

Substituting \eqref{e:fixedTexp} into the full Yukawa tensor, including the fundamental-period factor, gives a series in $y$. Inverting the mirror map expresses this series in $Q_2$, and the coefficient of $Q_2^2$ gives
\begin{equation}\label{e:D2J}
 D_2(Q)+\frac18D_1(Q^2)=\frac{2e^{-2g}}{P^2}J,
\end{equation}
where
\begin{equation}\label{e:Jdef}
 J=\frac{6x^2-8x+3}{2s^3}+\frac{k}{2s^2}
 +\frac{3(2x-1)k}{xs^2}+\frac{6k^2}{x^2s}
 -\frac{2h}{s}-\frac{kL+x^2V_P/4}{s}.
\end{equation}
The term $-2h/s$ combines the corrections to the three tangent-vector factors and to the inverse mirror map. The numerical factor in \eqref{e:D2J} follows from $[Q_2^2]\mathcal Y_{222}=8D_2(Q)+D_1(Q^2)$ in \eqref{e:Y222definition}.

It remains to express $J$ in modular variables. Substituting \eqref{e:modmirror} and the Ramanujan identities \eqref{e:Ramanujan}, with $\vartheta x=xE_6/E_4$, reduces $s^2E_4J/2$ to the bracket in \eqref{e:D2}; the terms linear in $\vartheta g$ cancel. The derivatives are obtained from $L=(\vartheta E_4)/(2E_4\vartheta x)$ and $V_P=L'+L^2$. This proves the second formula. Both formulas follow from the Picard--Fuchs recurrence and the full Yukawa tensor, without a modular ansatz or input from the gauge-instanton coefficients.
\end{proof}

The first instanton numbers in this curve-class basis are listed below. The entries are $n^0_{j,k}$; the rational invariants $N^0_{j,k}$ also include multiple covers. For example, \eqref{e:GWmulticover} gives $N^0_{0,2}=1/4$, whereas $n^0_{0,2}=0$.
\begin{center}
\begin{tabular}{r rr}
\toprule
$j$&$n^0_{j,1}$&$n^0_{j,2}$\\\midrule
0&2&0\\
1&2496&0\\
2&1941264&223752\\
3&1327392512&1327392512\\
4&861202986072&2859010142112\\
5&540194037151104&4247105405354496\\\bottomrule
\end{tabular}
\end{center}
The values for $0\le j\le4$ agree with \cite[Table~8.2]{Candelas}, and the row at $j=5$ agrees with the table following equation~(5.3) in \cite{HM}. The conjecture in \cite{HM} on modularity at every instanton order is not used in the derivation.

\subsection{Singularities and asymptotics of the instanton numbers}
\begin{theorem}[Asymptotics at large fiber degree]\label{e:growththm}
Define
\begin{equation}\label{e:kappas}
 \kappa_1=\frac{8e^{-g_*}}{\mathcal E^2},\qquad
 \kappa_2=\frac{e^{-2g_*}}{6\mathcal E^4}.
\end{equation}
Then
\begin{align}\label{e:growth}
 n^0_{j,1}&=e^{2\pi j}\bigl[\kappa_1j+O(\log j)\bigr],\notag\\
 n^0_{j,2}&=e^{2\pi j}\bigl[\kappa_2j^5+O(j^4\log j)\bigr].
\end{align}
In particular $\kappa_2/\kappa_1^2=1/384$, proving \eqref{e:countlimit}.
\end{theorem}
\begin{proof}
Write $s=1-x$ and work on a fixed slit sector at $s=0$. The local hypergeometric expansions give
\begin{align}\label{e:localbounds}
 P&=P(1)+O(s^{1/2}),& P'&=O(s^{-1/2}),&P''&=O(s^{-3/2}),\notag\\
 g&=g_*+O(s^{1/2}\log s),&g'&=O(s^{-1/2}\log s),&g''&=O(s^{-3/2}\log s).
\end{align}
The logarithm in the bounds for $g$ can be seen by varying a hypergeometric parameter. If
\[
 P_\nu(x)={}_3F_2(\tfrac16,\tfrac12,\tfrac56;1,\nu;x),
 \qquad G(x)=\sum_{m\ge1}a_mH_mx^m,
\]
then $G=-\left.\partial_\nu P_\nu\right|_{\nu=1}$ and $g=2G/P$, with $\nu$ near $1$. The indicial exponents at $x=1$ are $0,1,\nu-1/2$. The analytic Frobenius recurrence leaves two initial coefficients free. For $n\ge2$, the coefficient of the new term is proportional to $n(n-1)(n-\nu+1/2)$ and is nonzero near $\nu=1$. Hence two local solutions are analytic, while the third is $s^{\nu-1/2}$ times an analytic function.

This local basis and the coefficients expressing $P_\nu$ in it depend holomorphically on $\nu$. Differentiating at $\nu=1$ produces $s^{1/2}\log s$, with no more singular term. Frobenius convergence on smaller slit sectors \cite[Section~16.8]{DLMF} and $P(1)\ne0$ give the first estimates in \eqref{e:localbounds}. Termwise differentiation of the convergent expansions gives the derivative bounds; equivalently, they follow from Cauchy estimates on a slightly larger fixed sector.

The bounds imply $k=-1/(8s)+O(s^{-1/2})$ and $h=O(s^{-3/2}\log s)$. The coefficient of $s^{-3}$ in \eqref{e:Jdef} comes from four terms:
\begin{equation}\label{e:fourterms}
 \frac12-\frac1{16}-\frac38+\frac3{32}=\frac5{32}.
\end{equation}
All remaining terms are $O(s^{-5/2}\log s)$. Consequently,
\begin{equation}\label{e:Dsingx}
 D_1\sim\frac{2e^{-g_*}}{\mathcal E s},\qquad
 D_2\sim\frac{5e^{-2g_*}}{16\mathcal E s^3}.
\end{equation}
Both expressions have relative error $O(s^{1/2}\log s)$. The multiple-cover term $D_1(Q^2)/8$ is holomorphic at $Q=\rho$, so it does not affect the leading singularities.

Ramanujan's identities give $\left.(dE_6/dT)\right|_{T=i}=-\pi i\mathcal E^2$. With $z=T-i$ and $\delta=1-Q/\rho$,
\begin{equation}\label{e:criticalcoords}
 s=-\pi^2\mathcal E z^2+O(z^3),\qquad
 \delta=-2\pi iz+O(z^2),\qquad
 s=\frac{\mathcal E}{4}\delta^2+O(\delta^3).
\end{equation}
Consequently
\begin{equation}\label{e:Dsingdelta}
 D_1=\kappa_1\delta^{-2}[1+O(\delta\log\delta)],\qquad
 D_2=5!\kappa_2\delta^{-6}[1+O(\delta\log\delta)].
\end{equation}

To pass from these local expansions to coefficient asymptotics, we also need to know that $Q=\rho$ is the unique singularity of smallest modulus and that the functions continue to a suitable slit neighborhood of $|Q|=\rho$. Consider the operator
\begin{equation}\label{e:inhomogeneousPF}
 \mathcal L_0=\theta_x^3-x(\theta_x+\tfrac16)(\theta_x+\tfrac12)(\theta_x+\tfrac56).
\end{equation}
The coefficient recurrence and $H_m-H_{m-1}=1/m$ give
\begin{equation}\label{e:inhomogeneousidentity}
 \mathcal L_0P=0,\qquad \mathcal L_0G=\theta_x^2P.
\end{equation}
The coupled system \eqref{e:inhomogeneousidentity} has singularities only at $x=0,1,\infty$. Choose $R$ so that $\rho<R<e^{-\pi\sqrt3}$. By the valence formula, the zeros of $E_4$ and $E_6$ lie in the modular orbits of $e^{2\pi i/3}$ and $i$, respectively \cite[Proposition~2]{Zagier}. Every image of $e^{2\pi i/3}$ has imaginary part at most $\sqrt3/2$, so $E_4$ has no zero in $|Q|<R$. It follows that $P^2=E_4$ is nonzero there. Since the modular discriminant $\Delta_{\mathrm{mod}}=(E_4^3-E_6^2)/1728$ has no zero in the upper half-plane, $x(Q)=1728\Delta_{\mathrm{mod}}/E_4^3$ has no pole in this disk and vanishes only at $Q=0$.

An image of $i$ under a modular transformation with lower row $(c,d)$ has imaginary part $1/(c^2+d^2)$. The maximum is one, and every image attaining it gives $Q=\rho$; all other values are at most $1/2$. Thus $x(Q)=1$ only at $Q=\rho$ in $|Q|<R$.

The pullbacks $P(x(Q))$ and $G(x(Q))$ have holomorphic germs at $Q=0$, so the cusp is removable for both functions. The differential system continues them to the disk slit along $[\rho,R)$. Since $P(x(Q))$ is nonzero, $g(x(Q))=2G(x(Q))/P(x(Q))$ is holomorphic on the same domain. Also, $R<\sqrt\rho$, so $D_1(Q^2)$ is holomorphic throughout the disk.

We must account for the quadratic ramification of $x(Q)$ at $Q=\rho$ when transferring the local expansions. On the branch continued from $0<Q<\rho$,
\[
 \sqrt{s}=\frac{\sqrt{\mathcal E}}2\,\delta\bigl(1+O(\delta)\bigr),\qquad
 \log s=2\log\delta+\log(\mathcal E/4)+O(\delta).
\]
The units in these expressions are holomorphic near $\delta=0$. The formulas remain valid on the marked branch when $s$ crosses a principal-branch cut. Consequently, \eqref{e:Dsingdelta} holds uniformly as $Q\to\rho$ in the indented domain
$|Q|<R$, $|\arg(Q-\rho)|>\eta$, for any fixed $0<\eta<\pi/2$.
The continuation away from $\rho$ therefore gives the domain required for singularity analysis \cite[Theorem~2]{FO}. For a positive integer $p$,
\[
 [Q^j](1-Q/\rho)^{-p}=\rho^{-j}\binom{j+p-1}{p-1}.
\]
Taking $p=2$ and $p=6$ gives the leading terms $\kappa_1\rho^{-j}j$ and $\kappa_2\rho^{-j}j^5$. The logarithmic remainders in \eqref{e:Dsingdelta} contribute $O(\rho^{-j}\log j)$ and $O(\rho^{-j}j^4\log j)$, respectively, proving \eqref{e:growth}.
\end{proof}

For the rigid-limit comparison, we also need the leading singular coefficients in the coordinate $T-i$:
\begin{align}\label{e:residues}
 r_1&:=\lim_{T\to i}(T-i)^2D_1(e^{2\pi iT})
      =-\frac{2e^{-g_*}}{\pi^2\mathcal E^2},\notag\\
 r_2&:=\lim_{T\to i}(T-i)^6D_2(e^{2\pi iT})
      =-\frac{5e^{-2g_*}}{16\pi^6\mathcal E^4}.
\end{align}
The signs refer to the complex coordinate $T-i$. Along the real radial approach $Q\to\rho^-$, both leading terms are positive. These coefficients satisfy $r_2/r_1^2=-5/(64\pi^2)$.

\section{From Gromov--Witten invariants to Schwarzian coefficients}\label{sec:enumerative-rigid}
Turning the compact asymptotics into Schwarzian coefficients requires control of both the normalization and the rigid limit. We obtain the electric-coordinate and transverse normalization constants from locally uniform period limits, then prove convergence of the rescaled Yukawa coupling on a common holomorphic domain. Cauchy's coefficient formula extracts the coefficient at each fixed base degree, giving the first two gauge-instanton coefficients and the asymptotic relation \eqref{e:directrelation}.

\subsection{Limits of the flat coordinates and normalization}
Set $z=T-i$ and retain the base coordinate $\Sigma$. Subtracting the polynomial $\Sigma(2iz-1)$, of degree at most two in these coordinates, rewrites the base-dependent classical term $\Sigma T^2$ as $\Sigma z^2$. The classical cubic term of base degree zero remains in $f_0(T)$. The subtraction is an allowed constant complex symplectic change of the local special-coordinate frame, with no requirement that it preserve the integral charge lattice. The functions $F_k$ remain unchanged, and we make no further instanton-dependent change of the base coordinate.

For $k=1,2$, define the leading compact singular coefficient by $c_k=\lim_{z\to0}z^{4k-2}F_k(i+z)$ on the chosen slit sector. Thus $F_k$ has leading term $c_kz^{2-4k}$ and may contain logarithms at subleading orders. These local coefficients $c_k$ are distinct from the elliptic Frobenius coefficients in \eqref{eq:w0}. Since the multiple-cover term is regular here, \eqref{e:cover} gives
\begin{equation}\label{e:ck}
 c_1=\frac{r_1}{(2\pi i)^3}
      =-\frac{i e^{-g_*}}{4\pi^5\mathcal E^2},\qquad
 \frac{c_2}{c_1^2}=\frac{5\pi i}{8}.
\end{equation}
Thus both constants in \eqref{e:ck} come from the compact generating functions; no gauge-instanton coefficient has entered the calculation.

\begin{proposition}[Limits of the flat coordinates]
\label{e:fixed-boundary-flat-coordinates}
Let $U$ be a coordinate disk with compact closure in $\{u\in\CC:\operatorname{Re}u<-1\}$. In the chart $x=1+\epsilon u$, $y=\epsilon^2/(1+\epsilon u)^2$, let $\epsilon\to0^+$. Continue the period marking from real $u<-1$, taking $\sqrt u=i\sqrt{-u}$ there. Then the large-radius flat coordinates satisfy
\begin{equation}\label{e:fixed-boundary-flat-limits}
 \frac{Q_2}{\epsilon^2}\longrightarrow\frac{e^{g_*}}4,
 \qquad
 \frac{T-i}{\sqrt\epsilon}\longrightarrow
 \frac{\sqrt2}{\pi\sqrt{\mathcal E}}\,a(u).
\end{equation}
The convergence is uniform on compact subsets of $U$, together with any fixed number of $u$ derivatives. These limits hold at fixed $u$, rather than only in the asymptotic region $u\to\infty$.
\end{proposition}
\begin{proof}
Begin with $\epsilon>0$ and real $u<-1$. The coefficients $a_m$ in \eqref{e:Pg} give the following averaging formula for the fundamental period:
\begin{equation}\label{e:period-average}
 \omega_{\mathrm{CY}}(x,y)=\frac1{2\pi}\int_0^{2\pi}
 P\bigl(x(1+\sqrt y\cos\theta)\bigr)\,d\theta.
\end{equation}
To check the identity, expand the integrand. The odd moments of $\cos\theta$ vanish, and its $2n$th moment is $(2n)!/[4^n(n!)^2]$. The coefficient of $x^my^n$ is consequently $a_m m!/[4^n(n!)^2(m-2n)!]$, as in the Frobenius series. In the boundary chart, the argument of $P$ becomes $1+\epsilon(u+\cos\theta)$. The function $P$ is continuous at $1$ on this branch, so dominated convergence gives $\omega_{\mathrm{CY}}\to P(1)$ uniformly on compact real intervals with $u<-1$.

Let $\Pi_\Sigma$ be the logarithmic Frobenius period, normalized so that $\Sigma=\Pi_\Sigma/\omega_{\mathrm{CY}}$. Write
\begin{equation}
 2\pi i\Pi_\Sigma=\omega_{\mathrm{CY}}\log(y/4)+\mathcal G(x,y).
\end{equation}
Here $\mathcal G$ is the holomorphic part of the logarithmic period. Differentiating with respect to the second Frobenius exponent gives two kinds of terms. For $m\ge2n$, the contribution is
\begin{equation}\label{e:base-regular-coefficients}
 2a_m\frac{m!}{(n!)^2(m-2n)!}
 (H_{m-2n}-H_n)x^m(y/4)^n.
\end{equation}
When $m<2n$, the fundamental-period coefficient vanishes, but the derivative of its reciprocal Gamma factor does not. These terms give the tail
\begin{equation}\label{e:base-tail-coefficients}
 2\sum_{n\ge1}\sum_{m=0}^{2n-1}
 a_m(-1)^m\frac{m!(2n-m-1)!}{(n!)^2}x^m(y/4)^n.
\end{equation}
Although absent from the fundamental period, this tail contributes to the logarithmic period and must be included.

For sufficiently small $\epsilon$ on the real sector, $0<x<1$ and $x(1+\sqrt y)=1+\epsilon(u+1)<1$. When $2n\le m$, the inequality $|H_{m-2n}-H_n|\le H_m$ bounds the sum of the absolute values of the regular terms at fixed $m$ by
\[
 2a_mH_m[x(1+\sqrt y)]^m\le2a_mH_m.
\]
The majorant is summable. For each fixed $m$, only the term with $n=0$ survives as $\epsilon\to0$, and it tends to $2a_mH_m$. The inequalities $a_m\le1$, $x\le1$, and $m!(2n-m-1)!\le(2n-1)!$ bound the tail by
\[
 2\sum_{n\ge1}(y/4)^n\frac{2n(2n-1)!}{(n!)^2}
 \le\frac{2y}{1-y}.
\]
Thus dominated convergence gives
\begin{equation}\label{e:base-finite-part}
 \mathcal G(x,y)\longrightarrow2\sum_{m\ge1}a_mH_m=P(1)g_*,
 \qquad
 \log Q_2-\log(y/4)\longrightarrow g_*.
\end{equation}
Since $y/\epsilon^2\to1$, this proves the first limit in \eqref{e:fixed-boundary-flat-limits} for real $u<-1$.

For the electric coordinate, put $\Pi_T=\omega_{\mathrm{CY}}T(x,y)$. At $y=0$, this is the product $P(x)T(x)$, which we denote by $PT$. The electric period satisfies the analogous averaging formula
\begin{equation}\label{e:electric-average}
 \Pi_T(x,y)=\frac1{2\pi}\int_0^{2\pi}
 (PT)\bigl(x(1+\sqrt y\cos\theta)\bigr)\,d\theta.
\end{equation}
Expand both sides in $y$ at fixed $x$. The coefficient of $y^n$ is
\[
 \frac{x^{2n}(PT)^{(2n)}}{4^n(n!)^2},
\]
where the superscript denotes the $2n$th derivative in $x$. The same coefficients follow successively from $\mathcal L_2$. Since the coefficient $n^2$ in the recurrence is nonzero for $n\ge1$, the period at $y=0$ determines the solution holomorphic in $y$ uniquely.

Integrating \eqref{e:Tprime} from $x=1$, on the chosen branch, gives
\[
 T(x)-i=\frac{i\sqrt s}{\pi P(1)}+O(s),\qquad
 P(x)\bigl(T(x)-i\bigr)=\frac{i\sqrt s}{\pi}+O(s),
 \qquad s=1-x.
\]
Substitute this expansion into the electric average and divide by the fundamental-period average to obtain
\begin{equation}\label{e:electric-full-limit}
 \frac{T-i}{\sqrt\epsilon}
 \longrightarrow\frac{i}{\pi P(1)}\,
 \frac1{2\pi}\int_0^{2\pi}\sqrt{-u-\cos\theta}\,d\theta.
\end{equation}
On the same real sector, the quartic electric period in our marking is
\begin{equation}\label{e:electric-average-SW}
 a(u)=\frac{i}{\sqrt2}\,
 \frac1{2\pi}\int_0^{2\pi}\sqrt{-u-\cos\theta}\,d\theta.
\end{equation}
Expanding the integral for $|u|>1$ on the branch $\sqrt u=i\sqrt{-u}$ gives
\[
 a(u)=\sqrt{u/2}\,{}_2F_1(-1/4,1/4;1;u^{-2}),
\]
which agrees with \eqref{eq:periodnorm}. Comparison with \eqref{e:electric-full-limit} now proves the second limit in \eqref{e:fixed-boundary-flat-limits}, including its phase.

To prove uniform convergence for complex arguments, choose a slightly larger disk containing $\overline U$ and still lying in $\{\operatorname{Re}u<-1\}$. For all sufficiently small positive $\epsilon$, uniformly on its closure,
\begin{equation}\label{e:local-uniform-domain}
 |x|<1,\qquad |x|(1+\sqrt{|y|})=|1+\epsilon u|+\epsilon<1.
\end{equation}
With absolute values inserted, the same summable majorant $2a_mH_m$ bounds the regular part of the logarithmic period, and $2|y|/(1-|y|)$ bounds the tail. Each fixed-$m$ contribution has a uniform limit, which proves \eqref{e:base-finite-part} uniformly on the larger disk. In the averaging integrals, $1-x(1+\sqrt y\cos\theta)=-\epsilon(u+\cos\theta)$ stays in a fixed sector avoiding the negative real axis. The local expansions of $P(x)$ and $P(x)(T(x)-i)$ therefore hold uniformly under the integrals. They give uniform convergence of $\omega_{\mathrm{CY}}$ to $P(1)$ and of the rescaled electric period in \eqref{e:electric-full-limit}. The limiting fundamental period is nonzero, so division proves both limits. Cauchy estimates on the larger disk give convergence of every fixed-order $u$ derivative on $U$.

The period comparison thus holds without any additional limit in $u$.
\end{proof}

The electric-coordinate and transverse normalization constants are therefore
\begin{equation}\label{e:scales}
 z=\alpha\sqrt\epsilon\,a+o(\sqrt\epsilon),\qquad
 Q_2=\frac{\epsilon^2}{\beta_{\mathrm n}}(1+o(1)),\qquad
 \alpha^2=\frac2{\pi^2\mathcal E},\quad\beta_{\mathrm n}=4e^{-g_*}.
\end{equation}
The phase in \eqref{e:fixed-boundary-flat-limits} fixes $\alpha=\sqrt2/(\pi\sqrt{\mathcal E})>0$. The subscript in $\beta_{\mathrm n}$ distinguishes the scalar transverse normalization from the $M$-valued one-form $\beta$ of Section~\ref{sec:comparison}. Both limits are locally uniform in the sector of Proposition~\ref{e:fixed-boundary-flat-coordinates}.

The same tensor determines the perturbative singularity. Take its $TTT$ component and let $y\to0$, differentiating in $T$ at fixed flat $\Sigma$. This gives
\begin{equation}\label{e:oneloopYukawa}
 f_0'''(T)=\frac{4P}{\sqrt{1-x}}-6x\sqrt{1-x}\,P g'(x)
 \sim\frac{4i}{\pi(T-i)}.
\end{equation}
The estimate \eqref{e:localbounds} controls the second term. On the chosen branch,
\[
 \sqrt{1-x}\sim-\pi i\sqrt{\mathcal E}\,(T-i).
\]
More precisely, $f_0'''(T)-4i/[\pi(T-i)]=O(\log(T-i))$ on a smaller slit sector. Integration gives
\[
 f_0(i+z)=\frac{2i}{\pi}z^2\log z+d_0+d_1z+d_2z^2
             +O(z^3\log z)
\]
for constants $d_0,d_1,d_2$. Remove the constant and linear terms using the quadratic ambiguity of the homogeneous prepotential before rescaling. Since
\[
 \Sigma=\frac{\log\epsilon}{\pi i}
       -\frac{\log\beta_{\mathrm n}}{2\pi i}+o(1),
 \qquad \frac1{\pi i}+\frac12\frac{2i}{\pi}=0,
\]
the divergent logarithms cancel in the coefficient of $z^2$ before we take the limit:
\[
 \Sigma+\frac{2i}{\pi}\log z
 =\frac{2i}{\pi}\log\!\left(\frac z{\sqrt\epsilon}\right)
   -\frac{\log\beta_{\mathrm n}}{2\pi i}+o(1).
\]
Since $z/\sqrt\epsilon\to\alpha a$, dividing the potential by $\alpha^2\epsilon$ gives the rigid logarithmic coefficient $2i/\pi$.

The Yukawa coupling leaves a finite quadratic ambiguity in the prepotential. We fix it using the magnetic-period normalization of Section~\ref{sec:sw}, which also fixes the convention for $q=e^{\pi i\tau}$. This choice does not affect the coefficients of $a^{-2}$ and $a^{-6}$. We therefore compare coefficients with the gauge-theory prepotential
\begin{equation}\label{e:SWnorm}
 \mathcal F_{\rm SW}=\frac{ia^2}{2\pi}
 \left[2\log a^2-6+8\log2-\sum_{n\ge1}f_na^{-4n}\right].
\end{equation}
\subsection{The normalized rigid limit of the Yukawa coupling}
To extract the first two coefficients, we need a common holomorphic domain. Choose a negative real number $u_0$ with sufficiently large $|u_0|$, and let $A$ be a small disk about $\sqrt{u_0/2}$ on the chosen branch, disjoint from $a=0$. Let $\Lambda$ denote the dynamical scale and set $\mu=\Lambda^4$. Introduce this parameter through the flat-coordinate equations
\begin{equation}\label{e:coefficientfamily}
 T=i+\alpha\sqrt\epsilon\,a,\qquad
 Q_2=\epsilon^2\mu/\beta_{\mathrm n},\qquad
 C_\epsilon(a,\mu)=\alpha\sqrt\epsilon\,
 \partial_T^3\mathscr F_{\mathrm{CY}}(T,\Sigma).
\end{equation}
The derivatives in $T$ are taken at fixed $\Sigma$, so they remove the classical term $\Sigma T^2$ and every quadratic ambiguity. The compact expansion \eqref{e:CYF}, continued to the chosen branch, extends $C_\epsilon$ holomorphically across $\mu=0$ for each small positive $\epsilon$. We will also verify this extension directly from the regularized tensor on a common domain.
For $\mu\ne0$, restore the dynamical scale by setting
\begin{equation*}
 \mathcal F_{\rm SW}(a,\mu)
 =\mu^{1/2}\mathcal F_{\rm SW}(a/\mu^{1/4})
 \quad\text{modulo terms quadratic in }a.
\end{equation*}
Choose consistent branches and use \eqref{e:SWnorm}, where the dynamical scale is set to one. The only logarithmic term in $\mu$ is $-ia^2\log\mu/(2\pi)$, and three derivatives in $a$ remove it. Thus $C_{\mathrm{SW}}:=\partial_a^3\mathcal F_{\mathrm{SW}}(a,\mu)$ is holomorphic at $\mu=0$, and its $n$th instanton term is proportional to $\mu^n$.

We obtain a domain in $\mu$ independent of $\epsilon$ by introducing an auxiliary algebraic parameter $v$ and writing
\[
 x=1+\epsilon u,\qquad y=\epsilon^2v/x^2
\]
near $(u_0,0)$. On a sufficiently small complex polydisk in $(u,v)$, uniformly for small positive $\epsilon$,
\[
 |x|(1+\sqrt{|y|})=|1+\epsilon u|+\epsilon\sqrt{|v|}<1.
\]
The period averages and logarithmic-series estimates from the preceding proof apply uniformly here. The averages are even in $\sqrt v$, and the regularized logarithmic period is given by its convergent series. Thus the fundamental, electric, and regularized logarithmic periods are holomorphic in $(u,v)$. Dominated convergence and Cauchy estimates on a smaller polydisk give
\begin{equation}\label{e:uniformflatmap}
 \left(\frac{T-i}{\sqrt\epsilon},\frac{Q_2}{\epsilon^2}\right)
 \longrightarrow\left(\alpha a_v(u),\frac v\beta_{\mathrm n}\right),
 \qquad
 a_v(u)=\frac1{\sqrt2}\frac1{2\pi}
       \int_0^{2\pi}\sqrt{u+\sqrt v\cos\theta}\,d\theta.
\end{equation}
The square root is continued from the branch in \eqref{e:electric-average-SW}. For small $|v|$, expansion under the integral shows that $a_v$ is holomorphic in $v$: the odd powers of $\sqrt v$ have zero average. At $v=0$, we have $a_0(u)=\sqrt{u/2}$, and the limiting map in \eqref{e:uniformflatmap} has an invertible Jacobian at $(u_0,0)$.

The maps and their first derivatives converge uniformly on a slightly larger polydisk. To choose inverse domains independent of $\epsilon$, compose on the right with the inverse of the limiting map. The resulting maps have the form $\mathrm{id}+R_\epsilon$, with $R_\epsilon\to0$ in $C^1$ on a fixed neighborhood. On a smaller closed ball, $\|dR_\epsilon\|<1/2$ for all sufficiently small $\epsilon$. Shrink the target ball so that, for each target point $\mathbf z$, the equation $\mathbf w=\mathbf z-R_\epsilon(\mathbf w)$ defines a contraction of the source ball. It gives a unique inverse on this fixed target, and the inverse function theorem shows that the inverse is holomorphic. After reducing $A$ if necessary, we therefore obtain an $r>0$ independent of $\epsilon$ for which \eqref{e:coefficientfamily} is defined on $A\times\{|\mu|<r\}$. Choose the neighborhoods small enough that $u^2-v$ stays nonzero.

\begin{proposition}[Rigid Yukawa coupling in normalized coordinates]
\label{e:jacobiangauge}
On the common domain just constructed, $C_\epsilon(a,\mu)$ converges locally uniformly as $\epsilon\to0^+$ to $\partial_a^3\mathcal F_{\mathrm{SW}}(a,\mu)$. The limit includes $\mu=0$. In taking this limit, all mixed components of the compact Yukawa tensor are included.
\end{proposition}
\begin{proof}
Hold $\epsilon$ fixed throughout the computation, and write the rescaled flat coordinates as
\begin{equation}\label{e:regularflatcoordinates}
 A_\epsilon=\frac{T-i}{\alpha\sqrt\epsilon},\qquad
 \mu_\epsilon=\frac{\beta_{\mathrm n} Q_2}{\epsilon^2}
             =v\mathcal B_\epsilon,\qquad
 \mathcal B_\epsilon=\frac{\beta_{\mathrm n} Q_2}{\epsilon^2v}.
\end{equation}
The function $\mathcal B_\epsilon$ is a holomorphic unit, that is, a nowhere-vanishing holomorphic function. The simple zero of $Q_2$ along $y=0$ ensures that this unit extends across $v=0$. More explicitly, the regularized logarithmic period gives
\begin{equation}
 \mathcal B_\epsilon=x^{-2}
 \exp\!\left(\mathcal G/\omega_{\mathrm{CY}}-g_*\right).
\end{equation}
On smaller polydisks, the limits \eqref{e:uniformflatmap} imply
\begin{equation}\label{e:jacobianlimits}
 A_\epsilon\longrightarrow a_v(u),\qquad
 \mathcal B_\epsilon\longrightarrow1,\qquad
 \omega_{\mathrm{CY}}\longrightarrow P_*=\sqrt{\mathcal E},
\end{equation}
with convergence of every fixed-order derivative. The middle limit also holds at $v=0$. Indeed, $\mu_\epsilon-v$ vanishes there, so holomorphic division by $v$ and Cauchy estimates give $(\mu_\epsilon-v)/v\to0$ uniformly on a smaller polydisk.

The Jacobian matrix of these rescaled coordinates is
\begin{equation}\label{e:explicitJacobian}
 \frac{\partial(A_\epsilon,\mu_\epsilon)}{\partial(u,v)}
 =\begin{pmatrix}
 A_{\epsilon,u}&A_{\epsilon,v}\\
 v\mathcal B_{\epsilon,u}&\mathcal B_\epsilon+v\mathcal B_{\epsilon,v}
 \end{pmatrix}.
\end{equation}
The tangent direction at fixed flat base coordinate is described by
\begin{equation}\label{e:driftJacobian}
 \chi_\epsilon=-\frac{\mathcal B_{\epsilon,u}}
 {\mathcal B_\epsilon+v\mathcal B_{\epsilon,v}},\qquad
 V_\epsilon=\partial_u+v\chi_\epsilon\partial_v,\qquad
 J_\epsilon=V_\epsilon A_\epsilon.
\end{equation}
The vector field $V_\epsilon$ annihilates $\mu_\epsilon$. The determinant in \eqref{e:explicitJacobian} is $(\mathcal B_\epsilon+v\mathcal B_{\epsilon,v})J_\epsilon$, where $J_\epsilon$ is the derivative of the rescaled electric coordinate in this direction. Hence
\begin{equation}\label{e:fixedSigmaTangent}
 \left.\partial_T\right|_\Sigma
 =\frac{V_\epsilon}{\alpha\sqrt\epsilon\,J_\epsilon},
 \qquad \chi_\epsilon\longrightarrow0,
 \qquad J_\epsilon\longrightarrow\partial_u a_v\ne0.
\end{equation}
For $v\ne0$, fixing $\Sigma$ is equivalent to fixing $Q_2$. The formula extends the tangent direction holomorphically to $v=0$ and gives $(\partial v/\partial a)_\mu=v\chi_\epsilon/J_\epsilon$. To control the logarithmic tensor components, we need $\chi_\epsilon\to0$ itself; $v\chi_\epsilon\to0$ alone would not suffice.

We can now evaluate the full tensor on this tangent direction. Put
\begin{align}\label{e:JacobianPolynomials}
 \mathsf H_\epsilon&=1+2\epsilon u+\epsilon^2(u^2-v),\notag\\
 \mathsf P_\epsilon&=1+4\epsilon u+\epsilon^2(3u^2+v),\notag\\
 \mathsf R_\epsilon&=u+2\epsilon(u^2+v)+\epsilon^2(u^3+3uv).
\end{align}
Let $\xi=dv/v$ be the logarithmic one-form in the auxiliary parameter direction, and restrict $(2\pi i)^3\mathcal C$ to a fixed-$\epsilon$ slice, denoting the result by $\widehat{\mathcal C}_\epsilon$. In the logarithmic frame of \eqref{e:rationaltensor}, substitute
\begin{equation}
 X=\frac{\epsilon}{x}du,\qquad
 Y=\xi-\frac{2\epsilon}{x}du,\qquad
 \Delta=\epsilon^2(u^2-v).
\end{equation}
The full transformed tensor is
\begin{equation}\label{e:exactGeneralizedTensor}
 \frac{\widehat{\mathcal C}_\epsilon}{\epsilon}
 =\frac{4\mathsf P_\epsilon(du)^3
 -6\mathsf R_\epsilon(du)^2\xi
 +3v\mathsf P_\epsilon du\,\xi^2
 -\tfrac12v\mathsf R_\epsilon\xi^3}
 {(u^2-v)\mathsf H_\epsilon^2}.
\end{equation}
The numerical factors include the symmetric-tensor multiplicities, and every coefficient in this logarithmic frame is holomorphic at $v=0$. Evaluating with $du(V_\epsilon)=1$ and $\xi(V_\epsilon)=\chi_\epsilon$ gives
\begin{equation}\label{e:exactDriftContraction}
 \frac{\widehat{\mathcal C}_\epsilon(V_\epsilon,V_\epsilon,V_\epsilon)}{\epsilon}
 =\frac{4\mathsf P_\epsilon-6\mathsf R_\epsilon \chi_\epsilon
 +3v\mathsf P_\epsilon \chi_\epsilon^2
 -\tfrac12v\mathsf R_\epsilon \chi_\epsilon^3}
 {(u^2-v)\mathsf H_\epsilon^2}
 \longrightarrow\frac4{u^2-v}.
\end{equation}
The convergence is uniform on smaller polydisks, including at $v=0$. The mixed logarithmic terms therefore vanish in the limit along the direction at fixed flat base coordinate.

We must also pass to the period-normalized Hodge frame $\Omega^{\mathrm{norm}}=\Omega/\omega_{\mathrm{CY}}$. Write $\langle\ ,\ \rangle$ for the polarization pairing, defined by integrating wedge products over the Calabi--Yau fiber. Griffiths transversality and Hodge orthogonality give
\begin{equation}\label{e:hodgeGaugeExact}
 \langle\Omega,\nabla^m\Omega\rangle=0\quad(0\le m\le2),
 \qquad \mathcal C^{\rm norm}=\omega_{\mathrm{CY}}^{-2}\mathcal C.
\end{equation}
Here $\nabla^m$ denotes any iterated derivative of total order $m$ in the moduli directions. When we expand three derivatives of $\Omega/\omega_{\mathrm{CY}}$, every term containing a derivative of $1/\omega_{\mathrm{CY}}$ also contains one of the vanishing pairings above. Orthogonality likewise removes the lower-derivative terms from the coordinate change. The surviving factor is the tensor cube of its first derivative, as in the period normalization of \cite[Section~8.1]{Candelas}.

Together, \eqref{e:fixedSigmaTangent} and \eqref{e:hodgeGaugeExact} give the exact relation at finite $\epsilon$,
\begin{equation}\label{e:exactGaugeJacobian}
 C_\epsilon=
 \frac{\widehat{\mathcal C}_\epsilon(V_\epsilon,V_\epsilon,V_\epsilon)}
 {\alpha^2\epsilon(2\pi i)^3\omega_{\mathrm{CY}}^2J_\epsilon^3}.
\end{equation}
Here $\epsilon^{-1}$ extracts the leading Yukawa tensor, $J_\epsilon^{-3}$ accounts for the electric-coordinate change, and $\omega_{\mathrm{CY}}^{-2}$ normalizes the Hodge frame. The factor $(2\pi i)^{-3}$ reverses the rescaling in \eqref{e:rationaltensor}. Passing to the limit gives
\begin{equation}\label{e:explicitCubicLimit}
 C_\epsilon\longrightarrow
 \frac4{\alpha^2(2\pi i)^3\mathcal E
              (u^2-v)(\partial_u a_v)^3}
 =\frac{i}{4\pi(u^2-v)(\partial_u a_v)^3},
\end{equation}
where the normalization $\alpha^2\mathcal E=2/\pi^2$ gives the last equality.

For $v\ne0$, let $a_1=a$ be the period at unit dynamical scale and set $a_v(u)=v^{1/4}a_1(u/\sqrt v)$; define $a_{D,v}$ in the same way. The electric period agrees with the continued integral in \eqref{e:uniformflatmap}. With primes denoting $u$ derivatives at fixed $v$, their Wronskian and Picard--Fuchs equation take the form
\[
 a_v a_{D,v}'-a_{D,v}a_v'=\frac{i}{\pi},\qquad
 a_v''=-\frac{a_v}{4(u^2-v)},
\]
and the same differential equation holds for $a_{D,v}$. Consequently,
\begin{equation}\label{e:scaledSWCubic}
 \partial_a^3\mathcal F_{\rm SW}(a,v)
 =\frac{a_v'a_{D,v}''-a_{D,v}'a_v''}{(a_v')^3}
 =\frac{i}{4\pi(u^2-v)(a_v')^3}.
\end{equation}
The right-hand side agrees with \eqref{e:explicitCubicLimit}. At $v=0$, it extends to $4i/(\pi a)$ since $a_0(u)=\sqrt{u/2}$, and hence gives the perturbative Yukawa coupling. Finally, compose with the inverse coordinate maps on the common target domain. This replaces $a_v(u)$ by $a$ and $v$ by its limiting flat value $\mu$, while preserving locally uniform convergence. The proposition follows.
\end{proof}

\subsection{Extraction of gauge-instanton and Schwarzian coefficients}\label{sec:enumerative-coefficients}
Recall that $\mu$ is the fourth power of the dynamical scale. Fix a circle $|\mu|=r_0<r$ inside the common domain of Proposition~\ref{e:jacobiangauge}. By uniform convergence on this circle, we may pass to the limit in Cauchy's coefficient formula. For each $k=1,2$, we obtain
\begin{align}\label{e:finiteextraction}
 [\mu^k]C_{\rm SW}
 &=\lim_{\epsilon\to0}\alpha\sqrt\epsilon
       \left(\frac{\epsilon^2}{\beta_{\mathrm n}}\right)^k
       F_k'''(i+\alpha\sqrt\epsilon\,a)\notag\\
 &=\partial_a^3\left[
       \frac{c_k}{(\alpha^4\beta_{\mathrm n})^k}a^{2-4k}\right].
\end{align}
To prove the second equality, use the local expansions from Section~\ref{sec:enumerative}. For $k=1,2$, they give
\begin{align*}
 F_k(i+z)&=c_kz^{2-4k}+O\!\left(z^{3-4k}\log z\right),\\
 F_k'''(i+z)&=c_k(2-4k)(1-4k)(-4k)z^{-1-4k}
              +O\!\left(z^{-4k}\log z\right).
\end{align*}
The primes denote $T$ derivatives. Cauchy estimates on a smaller fixed slit sector justify differentiating the remainder. With $z=\alpha\sqrt\epsilon\,a$, the remainder contributes $O(\sqrt\epsilon\,|\log\epsilon|)$ to \eqref{e:finiteextraction}, uniformly for $a$ in compact subsets of $A$, and therefore tends to zero. A quadratic integration constant does not affect the coefficient of $a^{2-4k}$. Thus the leading compact term $c_kz^{2-4k}Q_2^k$ contributes $c_k/(\alpha^4\beta_{\mathrm n})^k$ to the rigid prepotential coefficient. The argument extracts each fixed base degree separately and requires no interchange of the rigid limit with an unbounded degree sum.

The normalization constants satisfy
\begin{equation}\label{e:scaleproduct}
 \alpha^4\beta_{\mathrm n}=\frac{16e^{-g_*}}{\pi^4\mathcal E^2}.
\end{equation}
It follows from \eqref{e:ck} that
\begin{equation}\label{e:fcomputed}
 -\frac{i}{2\pi}f_1=-\frac{i}{64\pi},\qquad
 -\frac{i}{2\pi}f_2=-\frac{5i}{32768\pi},
 \qquad
 f_1=\frac1{32},\quad f_2=\frac5{2^{14}}.
\end{equation}
This determines both gauge-instanton coefficients from the compact calculation. The factor involving $g_*$ cancels, so we do not need its numerical value.

The expansion in Section~\ref{sec:quantum}, with the normalization \eqref{e:SWnorm}, now gives the Schwarzian coefficients:
\begin{equation}\label{e:A12}
 A_1=32f_1,\qquad
 A_2=393216f_2-114688f_1^2.
\end{equation}
Substituting \eqref{e:fcomputed} gives $A_1=1$, $A_2=8$, and hence
\begin{equation}
 S(\tau)=32\pi^2q+256\pi^2q^2+O(q^3).
\end{equation}
We can also express the answer directly in terms of the compact asymptotic constants. Set $R=\kappa_2/\kappa_1^2$. Changing coordinates from $1-Q/\rho$ to $T-i$ in \eqref{e:Dsingdelta} gives
\begin{equation}
 r_1=-\frac{\kappa_1}{4\pi^2},\qquad
 r_2=-\frac{5!\kappa_2}{64\pi^6},\qquad
 \frac{r_2}{r_1^2}=-\frac{30}{\pi^2}R.
\end{equation}
Thus $c_2/c_1^2=240\pi iR$. Comparing coefficients of the negative powers of $a$ in \eqref{e:SWnorm} yields
\begin{equation}\label{e:factor120}
 \frac{f_2}{f_1^2}=120R.
\end{equation}
Divide the second identity in \eqref{e:A12} by $A_1^2$ to obtain $A_2/A_1^2=46080R-112$. The factor $5!$ comes from the coefficient asymptotic of the leading singular term $(1-Q/\rho)^{-6}$ and is essential to the normalization.

Finally, the multiple-cover relation \eqref{e:GWmulticover} gives
\begin{equation}\label{e:GW12}
 N^0_{j,1}=n^0_{j,1},\qquad
 N^0_{j,2}=n^0_{j,2}
 +\begin{cases}\frac18n^0_{j/2,1},&2\mid j,\\0,&2\nmid j.\end{cases}
\end{equation}
Theorem~\ref{e:growththm} bounds the additional term by $O(e^{\pi j}j)$, whereas $n^0_{j,2}\sim\kappa_2e^{2\pi j}j^5$. Also, $n^0_{j,1}\sim\kappa_1e^{2\pi j}j$ with $\kappa_1>0$, so the denominator below is nonzero for all sufficiently large $j$. We therefore have
\begin{equation}\label{e:GWlimit}
 \lim_{j\to\infty}\frac{e^{2\pi j}N^0_{j,2}}
 {j^3(N^0_{j,1})^2}
 =\frac{\kappa_2}{\kappa_1^2}=\frac1{384}.
\end{equation}
Together with $A_2/A_1^2=46080R-112$, this limit proves \eqref{e:directrelation}. At two fixed base degrees, the formula relates a nonlinear combination of Schwarzian coefficients to the leading asymptotics as the fiber degree tends to infinity. It makes no identification between an individual Fourier coefficient and an individual Gromov--Witten invariant.

At higher base degree, the same normalization calculation would give
\begin{equation}\label{e:conditionalhigher}
 \frac{f_n}{f_1^n}=(4n-3)!\,\frac{\kappa_n}{\kappa_1^n},
 \qquad n^0_{j,n}\sim\kappa_n e^{2\pi j}j^{4n-3},
\end{equation}
provided that the stated compact coefficient asymptotic holds and that its leading singular coefficient is identified with the corresponding rigid coefficient. We have established these hypotheses only for base degrees one and two. For $n\ge3$, \eqref{e:conditionalhigher} remains conditional on the additional singularity and coefficient estimates.

\clearpage
\section{Conclusion}
We have reconstructed the Seiberg--Witten Schwarzian quadratic differential from the leading logarithmic Yukawa tensor using a prescribed connection. The holomorphic germ of its coefficient in the effective coupling determines the normalized prepotential up to an additive constant. Using the genus-zero mirror correspondence, we have related the first two Schwarzian Fourier coefficients to the asymptotic growth of the compact invariants as the fiber degree tends to infinity. This comparison is established at base degrees one and two; extending it to higher degree requires further analytic estimates.

In the two-parameter mirror family, the leading tangential Yukawa coefficient is proportional to the holomorphic cubic tensor of rigid special K\"ahler geometry. The mixed normal components select a horizontal splitting along the open exceptional divisor. After imposing the constant-coefficient connection rule \eqref{g:connection} on the coefficient line bundle, local primitives recover the Seiberg--Witten Schwarzian quadratic differential
\[
 \Quad=-\frac{2(du)^2}{(u^2-1)^2}
      =\Sch{a_D/a}{\tau}(d\tau)^2.
\]
The differential is independent of the Hodge and normal frames used in the calculation. Of the two null lines, only one remains a horizontal complement at $u=0$, which makes the splitting unique on $B$.

Legendre uniformization expresses the coefficient of $\Quad$ in the effective coupling as the weight-four modular form $\pi^2\theta_2(\tau/2)^8/8$. Its holomorphic germ determines the prepotential up to an additive constant once the weak-coupling frame, dynamical scale, and perturbative normalization are fixed. In the normalized convergent class of Section~\ref{sec:quantum}, the Schwarzian therefore records all instanton corrections and vanishes identically precisely when they all vanish.

For the compact threefold, the generating functions at base degrees one and two give explicit asymptotics as the fiber degree tends to infinity. Uniform period estimates fix the electric-coordinate and transverse normalization constants. With these constants fixed, convergence of the full Yukawa tensor in flat coordinates justifies extraction of the gauge-instanton terms. We recover $f_1=1/32$ and $f_2=5/2^{14}$ and obtain \eqref{eq:enumerativeintro} as a relation among rational genus-zero Gromov--Witten invariants and Schwarzian coefficients.

The enumerative interpretation rests on the genus-zero mirror correspondence in the chosen chamber. The proof requires neither a comparison of integral variations nor a splitting of the limiting mixed Hodge structure. To extend the comparison to higher base degree, one must establish the additional singularity and coefficient estimates stated in Section~\ref{sec:enumerative-coefficients}. Until then, \eqref{e:conditionalhigher} remains conditional.

\acknowledgments
The author is grateful to Shamit Kachru and Arnav Tripathy for many helpful discussions and for reading an earlier draft. The author used ChatGPT (OpenAI) to assist with symbolic and numerical computations, editing, and checks of the mathematical arguments and formulas.

\clearpage
\bibliographystyle{plain}
\bibliography{Refs}
\end{document}